\documentclass[runningheads,orivec]{llncs}
\usepackage{pstricks}
\usepackage{multicol}
\usepackage{auto-pst-pdf}
\usepackage{amssymb}
\usepackage{amsmath}
\usepackage{mathrsfs}
\usepackage{stmaryrd}
\usepackage{times}
\usepackage{latexsym}
\usepackage{hhline}
\usepackage{enumerate} 
\usepackage{calc}     
\usepackage{graphicx}   
\usepackage{ifthen}    
\usepackage{pst-poly}   
\usepackage{multido}   
\usepackage[normalem]{ulem} 
\usepackage[textwidth=4cm,textsize=footnotesize]{todonotes}

\usepackage{eucal} 
\usepackage{arydshln}
\usepackage{MnSymbol}
\usepackage{float}

\makeatletter
\DeclareFontFamily{OMX}{MnSymbolE}{}
\DeclareSymbolFont{MnLargeSymbols}{OMX}{MnSymbolE}{m}{n}
\SetSymbolFont{MnLargeSymbols}{bold}{OMX}{MnSymbolE}{b}{n}
\DeclareFontShape{OMX}{MnSymbolE}{m}{n}{
	<-6> MnSymbolE5
	<6-7> MnSymbolE6
	<7-8> MnSymbolE7
	<8-9> MnSymbolE8
	<9-10> MnSymbolE9
	<10-12> MnSymbolE10
	<12->  MnSymbolE12
}{}
\DeclareFontShape{OMX}{MnSymbolE}{b}{n}{
	<-6> MnSymbolE-Bold5
	<6-7> MnSymbolE-Bold6
	<7-8> MnSymbolE-Bold7
	<8-9> MnSymbolE-Bold8
	<9-10> MnSymbolE-Bold9
	<10-12> MnSymbolE-Bold10
	<12->  MnSymbolE-Bold12
}{}

\let\llangle\@undefined
\let\rrangle\@undefined
\DeclareMathDelimiter{\llangle}{\mathopen}%
{MnLargeSymbols}{'164}{MnLargeSymbols}{'164}
\DeclareMathDelimiter{\rrangle}{\mathclose}%
{MnLargeSymbols}{'171}{MnLargeSymbols}{'171}
\makeatother

\usepackage{mathtools}

\begin{document}
	
	\title{Two-dimensional Dyck words}
	
	\author{Stefano {Crespi Reghizzi} 
		\inst{1} 
		\and 
		Antonio Restivo
		\inst{2} 
		\and 
		Pierluigi {San Pietro}
		\inst{1}
	}
	
	\authorrunning{S. {Crespi Reghizzi} \and A. Restivo \and P. {San Pietro}} 
	\institute{Politecnico di Milano - DEIB \and Dipartimento di Matematica e Informatica, Universit\`{a} di Palermo
		\email{stefano.crespireghizzi@polimi.it \quad antonio.restivo@unipa.it \quad pierluigi.sanpietro@polimi.it}}
	
	\maketitle      

\begin{abstract}
We propose different ways of lifting the notion of Dyck language from words to 2-dimensional (2D) arrays of symbols, i.e., pictures, by means of new definitions of increasing comprehensiveness.
Two of the proposals are based on alternative definitions of a Dyck language, which are equivalent over words but not on pictures.

\noindent First, the  property that any two pairs of matching parentheses are either well-nested or disjoint, is rephrased for rectangular boxes and leads to the well-nested Dyck, $DW_k$. The latter is a generalization of the known Chinese box language. We prove that, unlike the Chinese boxes, the language $DW_k$ is not recognizable by a tiling system.

\noindent Second, the Dyck cancellation rule is rephrased as a neutralization rule, mapping a quadruple of symbols representing the corners of a subpicture onto neutral symbols.The neutralizable Dyck language $DN_k$ is obtained by iterating neutralizations, starting from 2-by-2 subpictures, until the picture is wholly neutralized. 

\noindent Third, we define the Dyck crossword $DC_k$ as the row-column combination of Dyck word languages, which prescribes that each column and row is a Dyck word.
The relation between matching parentheses is represented in $DC_k$ by an edge of a graph situated on the picture grid. Such edges form
 a circuit, of path length multiple of four, of alternating row and column matches. Length-four circuits have rectangular shape, while longer ones exhibit a large variety of forms. A proper subset of $DC_k$, called quaternate, is also introduced by excluding all circuits of length greater than 4. 
We prove that $DN_k$ properly includes $DW_k$, and that it coincides with the quaternate $DC_k$ such that the neutralizability relation between subpictures induces a partial order. The 2D languages well-nested, neutralizable, quaternate and Dyck crossword are ordered by strict inclusions.
This work can be also seen as a first step towards the definition of context-free picture languages.
\end{abstract}

\section{Introduction}\label{sect:introd}
 The Dyck language is a fundamental concept in formal language theory. 
 Its alphabet $\{a_1,\ldots, a_k,\, a'_1,\ldots,a'_k\}$, for any $k \ge 1$, { is partitioned into the pairs $[a_1, a'_1], \ldots, [a_k, a'_k]$}. 
 The language is the set of all words that can be reduced to the empty word by cancellations of the form $a_i a'_i \to \varepsilon$. 
The centrality of the Dyck language is 
expressed by the Chomsky-Sch\"{u}tzenberger theorem~\cite{ChomskySchutz1963} stating that any context-free language is the homomorphic image of the intersection of a Dyck language and a local one; { intuitively, a regular language is local if it is defined by the set of factors, of prefixes and of suffixes of length two.}
\par
Motivated by our interest for the theory of two-dimensional (2D) or picture languages, we investigate the possibility to transport the Dyck concept from one dimension to 2D. 
When moving from 1D to 2D, most formal language concepts and relationships drastically change. In particular, in 2D the Chomsky's language hierarchy is blurred because the notions of regularity and context-freeness  cannot be formulated for pictures without giving up some characteristic properties that hold for words. In fact, it is known \cite{GiammRestivo1997} that the three equivalent definitions of regular languages by means of finite-state recognizer, by regular expressions, and by the homomorphism of local languages, produce in 2D three distinct language families. 
The third one gives the family of \emph{tiling system recognizable languages} (REC) \cite{GiammRestivo1997}, that many think to be the best fit for regularity in 2D.
\par
The situation is less satisfactory for context-free (CF) languages where a transposition in 2D remains problematic.  None of the existing proposals of ``context-free'' picture grammars (\cite{Siromoney&Subramanian&Rajkumar&Thomas:1999,Matz:1997,Nivat91SSSD,Prusa2004,DBLP:journals/tcs/Crespi-ReghizziP05,Drewes:2006}, a survey is \cite{DBLP:books/ems/21/Crespi-ReghizziGL21}) 
match the expressiveness and richness of formal properties of 1D CF grammars.
In this paper we make the first step towards a new definition of CF languages by means of the following 2D reformulation of the Chomsky-Sch\"{u}tzenberger theorem, that, { to avoid destroying the rectangular structure of a picture,} we take with a  non-erasing homomorphism, as in~\cite{DBLP:journals/fuin/BerstelB96,okhotin2012}. 
A context-free picture language is the homomorphic, letter-to-letter image of the intersection of a 2D Dyck language and a 2D local language. While the notion of 2D local language is well-known, 
we are not aware of any existing definitions of 2D Dyck language; we know of just one particular example, the \emph{Chinese box language} in~\cite{DBLP:journals/tcs/Crespi-ReghizziP05}, that intuitively consists of embedded or concatenated boxes, and was proposed to illustrate the expressiveness of the grammars there introduced; that language is not a satisfactory proposal, since it is in the family REC, hence ``regular''.
 Although a best definition of Dyck picture languages might not exist, it is worth formalizing and comparing several possible choices; this is our contribution, while the study of the resulting 2D CF languages is still under way and not reported here.
\par
Our contribution includes four definitions of 2D ``Dyck'' languages based on various approaches, a study of their properties and the proofs of their inclusions.
\par
It is time to describe the intuitions behind each proposal and, particularly, the properties of Dyck words that are preserved in each case.  
\par
Instead of open and closed parentheses, the elements of a Dyck alphabet in 2D are the four ``corners'' $\ulcorner,\,\urcorner, \,\llcorner, \, \lrcorner$; there may be $k\ge1$ distinct quadruples. Each corner quadruple encodes two 1D Dyck alphabets, one for rows the other for columns. { The row alphabet has the pairs [$\ulcorner,\,\urcorner$] and [$\llcorner, \, \lrcorner$] while the column alphabet has the pairs [$\ulcorner,\,\llcorner$] and [$\urcorner, \, \lrcorner$]}. 
When in a picture a corner quadruple is correctly laid on the four vertexes of a rectangular subpicture we say that it represents a rectangle. 
\par
We start from the simpler cases, the \emph{well-nested Dyck language} $DW_k$  and the \emph{neutralizable Dyck language} $DN_k$. In both cases, a picture is partitioned into rectangles, in the sense that each pixel is placed on a vertex of a rectangle. The difference between 
 $DW_k$  and $DN_k$  resides in the relative positions that are permitted for the rectangles that cover a picture. 
\par
In a $DW_k$  picture the rectangles are well nested and do not overlap each other; thus it is fair to say that the well-nesting property of parenthesized words is here preserved. This is the same constraint of the Chinese boxes~\cite{DBLP:journals/tcs/Crespi-ReghizziP05}, which however use a different alphabet that is not a Dyck alphabet.
\par
The definition of $DN_k$  is based on the observation that the Dyck cancellation rule can be replaced by a neutralization rule that maps a pair of matching parentheses onto a neutral symbol $N$, $a_i a'_i \to N N$, so that a word is well parenthesized if it can be transformed to a word in $N^+$ of the same length. 
In 2D the reformulation of the neutralization rule is: if a rectangle (having the four corner symbols as vertexes) includes only neutral symbols, then the whole subpicture is neutralized. A picture is in $DN_k$  if all corner symbols are replaced with $N$ by a sequence of neutralization steps.
We prove the language inclusion $DW_k \subset DN_k$.
\par
The third approach is based on the crosswords of Dyck languages, a.k.a. row-column compositions. A picture is in $DC_k$ if all rows and all columns are Dyck  words. Crosswords have been studied for regular languages (e.g., in \cite{DBLP:journals/tcs/LatteuxS97,DBLP:journals/iandc/FennerPT22}) but not, to our knowledge, for context-free ones.
A little reflection suffices to realize that in $DN_k$, hence also in $DW_k$, the rows and columns are Dyck words, therefore the inclusion $DN_k \subseteq DC_k$ is obvious. 
\par
The interesting question is whether the inclusion is strict. Surprisingly, $DC_k$ comprises pictures that do not belong to $DN_k$. A subclass of $DC_k$, called \emph{quaternate}, or $DQ_k$, is the set of pictures covered by rectangles. We prove that $DQ_k$ includes also not neutralizable pictures, which present a circularity in the precedence relation that governs the neutralization order. 
\par
{ But the family of Dyck crosswords includes a large spectrum of pictures where patterns other than rectangles are present. Each pattern is a closed path, called a \emph{circuit}, made by alternating horizontal and vertical edges, representing a Dyck match on a row or on a column. A circuit label is a string in the language
$(\ulcorner\;\urcorner \, \lrcorner\;\llcorner )^+$,
thus having length $4k$. The circuit path may intersect itself one or more times on the picture grid--the case of zero intersection is the rectangle. We prove that for any value of $k\ge 0$ there exist pictures in $DC_k$ featuring a circuit of length $4+8k$. We have examined some interesting types of Dyck crosswords that involve complex circuits, but much remains to be understood of the general patterns that are possible.} 
\par
Section~\ref{s-preliminaries} lists basic concepts of picture languages and Dyck languages.
Section~\ref{sectBoxBasedDyck} recalls the Chinese boxes language, defines the $DW_k$  and $DN_k$  languages, and studies their relations.
Section~\ref{sectDyckCrosswords} introduces the $DC_k$ languages, exemplifies the variety of circuits they may contain, and defines the quaternate subclass $DQ_k$. Section~\ref{s-inclusions} proves the strict inclusions of the four above languages. 
Section~\ref{s-conclusion} mentions open problems.

\section{Preliminaries}\label{s-preliminaries}

All the alphabets to be considered are finite.
The following concepts and notations for picture languages follow mostly~\cite{GiammRestivo1997}.
A {\em picture} is a rectangular array of letters over an alphabet.
Given a picture $p$, $|p|_{row}$ and $|p|_{col}$ denote the number of rows and columns, respectively; $|p|=\left(|p|_{row},|p|_{col}\right)$ denotes the \emph{picture size}. 
The set of all non-empty pictures over $\Sigma$ is denoted by $\Sigma^{++}$. 

\par\noindent
 A {\em domain} $d$ of a picture $p$ is a quadruple $(i,j,{i'}, {j'})$, with $1\le i \le i'\le |p|_{row}$, and $1\le j \le j'\le |p|_{col}$. The \emph{subpicture of $p$ } with domain $d=(i,j,{i'}, {j'})$, denoted by $ spic(p,d)$ is the (rectangular) portion of $p$ defined by the top-left coordinates $(i,j)$ and by the bottom right coordinates
$({i'}, {j'})$.

\par\noindent\emph{Concatenations.} Let $p,q \in \Sigma^{++}$. 
The {\em horizontal} {\em concatenation} of $p$ and $q$ is denoted as $p\obar q$ and it is defined
when $|p|_{row}= |q|_{row}$. 
Similarly, the {\em vertical} {\em concatenation} $p \ominus q$ is 
defined when $|p|_{col}= |q|_{col}$.
 We also use the power operations $p^{\ominus k}$ and $p^{\obar k}$, $k\geq 1$, their closures $p^{\ominus +},$ $p^{\obar +}$ and we extend the concatenations to languages in the obvious way. 
 
 The notation $N^{m,n}$, where $N$ is a symbol and $m,n>0$, stands for a homogeneous picture of size $m,n$. For later convenience, we extend this notation to the case where either $m$ or $n$ are 0, to introduce identity elements for vertical and horizontal concatenations: given a picture $p$ of size $(m,n)$, by definition $p \obar N^{m,0}= N^{m,0}\obar p = p$ 
 and $p \ominus N^{0,n}= N^{0,n}\ominus p = p$.

\par\noindent The \emph{Simplot closure}~\cite{DBLP:journals/tcs/Simplot99} operation $L^{**}$ is defined on a picture language~ $L$ as
the { set of pictures $p$ tessellated by pictures in $L$,} more precisely defined by the following condition:
\begin{gather}\nonumber
\exists \text{ a partition of } \{1, \ldots, row(p) \} \times \{1, \ldots, col(p) \} 
\text{ into } n\ge 1 \text{ domains } d_1, \ldots, d_n \text{ of } p
\\\label{eqDefSimplot}
\text{ such that for all } 1 \le i \le n \text{ the subpicture } spic(p, d_i) \text{ is in } L.
\end{gather}
Notice that the concatenations $L^{\ominus k}$, $L^{\obar k}$ and their closures $L^{\ominus +}$, $L^{\obar +}$ are included in the Simplot closure of $L$, which therefore is the most general way of assembling a picture starting from pictures in a given language. 
\par
{ To help understanding, we reproduce a picture in $L^{**}$, where $L$ is tessellated into the 1-by-1, 1-by-2, 2-by-1, and 2-by-2 pictures shown, and cannot be obtained by horizontal and vertical partitions: }
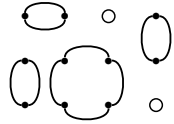
\begin{center}
\psset{arrows=-,labelsep=5pt,colsep=18pt,rowsep=12pt,nodealign=true}         
\scalebox{0.7}{
	\begin{psmatrix}                                      
		\dotnode{p11}{}&\dotnode{p12}{} &\circlenode{p13}{} &\dotnode{p14} {}      
		\\                                           
		\dotnode{p21}{} &\dotnode{p22}{} &\dotnode{p23}{} &\dotnode{p24}{}      
		\\                                           
		\dotnode{p31}{} &\dotnode{p32}{} &\dotnode{p33}{} &\circlenode{p34}{}       
		\\                                           
		
		\ncarc[arcangle=90,linestyle=solid,linewidth=1.0pt]{p11}{p12} 
		\ncarc[arcangle=90,linestyle=solid,linewidth=1.0pt]{p12}{p11} 
		\ncarc[arcangle=90,linestyle=solid,linewidth=1.0pt]{p14}{p24} 
    \ncarc[arcangle=90,linestyle=solid,linewidth=1.0pt]{p24}{p14} 		
		\ncarc[arcangle=90,linestyle=solid,linewidth=1.0pt]{p21}{p31} 
    \ncarc[arcangle=90,linestyle=solid,linewidth=1.0pt]{p31}{p21} 		

 \ncarc[arcangle=90,linestyle=solid,linewidth=1.0pt]{p22}{p23} 
 \ncarc[arcangle=90,linestyle=solid,linewidth=1.0pt]{p23}{p33} 
 \ncarc[arcangle=90,linestyle=solid,linewidth=1.0pt]{p33}{p32} 
 \ncarc[arcangle=90,linestyle=solid,linewidth=1.0pt]{p32}{p22}

%
%
%
	\end{psmatrix}
}
\end{center} 

\par{  We assume some familiarity with the basic properties of the family of REC languages~\cite{GiammRestivo1997}, in particular with their definition by the projection of a local 2D language or equivalently, by the projection of the intersection of two domino languages for rows and for columns.
}
\paragraph*{Dyck alphabet and language}
The definition and properties of Dyck languages are basic concepts in formal language theory, yet we prefer to list them since each one of our developments for 2D languages differs with respect to the property it strives to generalize.
	\par
	For a Dyck language $D \subseteq \Gamma_k^*$, the alphabet has size $|\Gamma_k|=2k$ and is partitioned into two sets of cardinality $k\ge 1$, 
	denoted $\{a_i \mid 1 \le i\le k\} \cup \{a'_i \mid 1 \le i\le k\}$. 
\par
	The Dyck language $D_k$ has several equivalent, definitions. We recall the word congruence or {\em cancellation rule} defined by $a_i a'_i = \varepsilon$: a word is in $D_k$ if it is congruent to $\varepsilon$, i.e., it can be erased to $\varepsilon$ by repeated application of the cancellation rule. 
	We say that in a word $x \in D_k$ two occurrences of terminals $a_i,\, a'_i$ \emph{match} if they are erased together by an application of the cancellation rule on $x$. 
\par
	A characteristic of Dyck languages is that, in every word of $D_k$, any two factors $a_i y a'_i$ and $a_j w a'_j$, with $y,w \in \Gamma^*$ where $a_i , a'_i$ and $a_j, a'_j$ are matching pairs, are either disjoint or \emph{well-nested}.
\par
	A (non-$\varepsilon$) Dyck word is {\em prime} if it is not the concatenation of two Dyck words. The set of prime words can be defined \cite{Berstel79} as the set 
	$(D_k-\varepsilon) -(D_k-\varepsilon)^2$. 
\par
	For future comparison with pictures, we introduce an equivalent definition of the Dyck language by means of the following \emph{neutralization rule} instead of the cancellation rule, since the latter does not work for pictures: erasing a subpicture would not result in a picture. 
	Let $N \notin \Gamma_k$ be a new terminal character called \emph{neutral}. 
		For every word in $(\Gamma_k\cup \{N\})^*$ define the congruence $\approx$, for all $1 \le i \le n$, and for all $m\ge 0$ as: 
\begin{equation}\label{eqNeutralization1DDyck}
a_i\,N^m a'_i \approx N^{m+2}.
\end{equation}
		A word $x \in\Gamma_k^*$ is in $D_k$ if it is $\varepsilon$ or it is $\approx$-congruent to $N^{|x|}$.
	An equivalent definition of the Dyck language is based on the observation that a Dyck word can be enlarged either by surrounding it with a matching pair of parentheses, or by concatenating it to another Dyck word. Therefore, the Dyck language over $\Gamma_k$ can be defined through a {\em nesting accretion} rule: 
	given a word $x\in \Gamma_k^*$, a nesting accretion of $x$ is a word of the form $a_i x a'_i$. The language $D_k$ can then be defined as the smallest set including the empty word and closed under concatenation and nesting accretion.

\section{Box-based choices of Dyck picture languages}\label{sectBoxBasedDyck}
In this section we present two simple choices, called well-nested and neutralizable, each one  conserving one of the characteristic properties of Dyck words. 
\par
To make the analogy with Dyck words more evident, we represent in 2D the  parentheses pair $[ \, ,\, ]$ by a quadruple of corners $\ulcorner,\,\urcorner, \,\llcorner, \, \lrcorner$. Then inside a picture such a quadruple matches if it is laid on the four vertexes of a rectangle (i.e., a subpicture), as we see in the picture 
$
\begin{array}{cccc}\red \ulcorner &\ulcorner& \urcorner&\red \urcorner
\\
\red \llcorner &\llcorner& \lrcorner& \red \lrcorner
\end{array}
$
for each quadruple identified by a color.

\par
First, we focus on the { nesting accretion definition of Dyck words and extend it to pictures by considering a quadruple of corners.} The corresponding picture languages are called \emph{well-nested Dyck}, denoted as $DW_k$.
Then, we extend the neutralization rule to 2D in a way that essentially preserves the following property: two matching parentheses that encompass a neutral word can be neutralized. Now, the two matching parentheses become a quadruple of symbols representing the corners of a (rectangular) subpicture already neutralized.  The corresponding languages are called \emph{neutralizable Dyck} ($DN_k$).

\subsection{Well-nested Dyck language}
The natural question of what should be considered a Dyck-like language in 2D received a tentative answer in \cite{DBLP:journals/tcs/Crespi-ReghizziP05}, where a language of well-embedded rectangular boxes, called \textit{Chinese boxes}, was presented as an example of the generative power of the tile rewriting grammars there introduced.
\par
The alphabet is $\Gamma = \{\ulcorner,\urcorner,\llcorner, \lrcorner, \bullet\}$; the corner symbols represent the four box vertexes and a horizontal/vertical string of bullets represents a box side.
Instead of the original grammar formalism, we give a recursive definition.
\begin{definition}[Chinese boxes \cite{DBLP:journals/tcs/Crespi-ReghizziP05}]\label{defChineseEmbedding}
	Given a picture $p$ of size $(n,m)$, with $n,m\ge 0$, its {\em Chinese accretion } is the picture: 
	\[ \left(\ulcorner \, \obar \, \bullet^{\obar m} \, \obar \, \urcorner \right) \, \ominus \,\left(	\bullet^{\ominus n} \, \obar\, p \,\obar \,\bullet^{\ominus n} \right) 
	\ominus \left(\llcorner \,\obar\, \bullet^{\obar m} \, \obar \, \lrcorner \right) 
	\]
	i.e., 
	\scalebox{0.7}{
		$
		\begin{array}{ccccc}
		\ulcorner & \multicolumn{3}{c}{\bullet \dots \bullet}& \urcorner
		\\ 
		\bullet & \multicolumn{3}{c}{ \quad}&\bullet
		\\
		\vdots&     \multicolumn{3}{c}{ p }&	 \vdots
		\\
		\bullet& \multicolumn{3}{c}{ \quad}& \bullet
		\\ 
		\llcorner & \multicolumn{3}{c}{ \bullet \dots \bullet}&\lrcorner
		\\
		\end{array}
		\quad
		$
	}
	or the picture $\begin{array}{cc} \ulcorner & \urcorner \\ \llcorner&\lrcorner
	\end{array}
	$
	when $p$ is empty.
	\par\noindent
	The {\em Chinese boxes language}, denoted by $DB$, is defined as the smallest set including the empty picture and closed under horizontal and vertical concatenations,
	and under Chinese accretion. 
\end{definition}
An example is in Figure~\ref{figChinese+Accretion+Wellnested} (middle).
Initially erroneously believed not to be recognizable, a tiling system for language $DB$ was later found (see \cite{DBLP:journals/pr/PradellaC08}) thus proving the contrary.
Therefore, if we agree with the opinion that the analogue of regular 1D languages are the recognizable picture languages, then $DB$ does not qualify as a good 2D analogue of the Dyck 1D languages since the latter are not regular.
\par{
In addition we observe that a row word is similar to a Dyck word since the pairs $[\ulcorner, \urcorner]$ and $[\llcorner, \lrcorner]$ act as matching pairs, but it also contains the symbol $\bullet$ which has no matching symbol\footnote{Such a matchless symbol, when is added to a Dyck alphabet, is usually called \emph{neutral}.}; the case for column words is analogous, with the matching pair $[\ulcorner, \llcorner]$ and $[\urcorner, \lrcorner]$.  
} 
\par
The next definition extends to 2D the following property of Dyck 1D: 
 a Dyck word can be enlarged either by surrounding it with a matching pair of parentheses, or by concatenating it to another Dyck word. Similarly, we are going to define the well-nested 2D Dyck languages by means of accretion of a well-nested picture, or by concatenations--actually by the more general Simplot closure.
\par
 The alphabet consists of the corner symbols, that 
 we prefer to represent by Latin letters to simplify typography, with the correspondence: 
$a=\,_{\big\ulcorner},\,b=\,_{\big\urcorner}, \,c=\llcorner, \, d=\lrcorner$, thus the picture \scalebox{0.8}{$\begin{array}{cc} \ulcorner & \urcorner \\ \llcorner&\lrcorner\end{array}$} is the same as \scalebox{0.8}{$\begin{array}{cc} a & b\\ c&d\end{array}$}.
In general, as for Dyck word languages, we allow multiple copies of the four symbols, i.e., the alphabet is $\Delta_k=\{a_i,b_i,c_i,d_i\mid 1\le i \le k\}$, $k\ge 1$, and $\Delta_1 = \{a,b,c,d\}$.

\par
The following definition is similar to Definition~\ref{defChineseEmbedding} but substitutes corners to the bullets ``$\bullet$''.

\begin{definition}[well-nested Dyck picture language]
	\label{defDyckWellNested} 
	Let $\Delta_k=\{a_i,b_i,c_i,d_i\mid 1\le i \le k\}$. Define two bijections: 
	$h_r: \{a_i,b_i\}\to \{c_i,d_i\}$, $h_c: \{a_i,c_i\}\to \{b_i,d_i\}$ 
	with $h_r (a_i)=c_i$, $h_r (b_i)=d_i$ and 
	$h_c (a_i)=c_i$, $h_c (c_i)=d_i$.
	
\par\noindent	For every picture $p \in \Delta_k^{++}$, for all rows $w_r$ in the (word) Dyck language over the parentheses ${[a_i, b_i]}$, and for all columns $w_c$ in the 
 Dyck language over the parentheses ${[a_i, c_i]}$, such that $|w_r|=|p|_{col}$, $|w_c|=|p|_{row}$, 
	the {\em nesting accretion} of $p$ within $w_r, w_c$ is the picture: 
	\[ \left(a_i \, \obar \, w_r\,\obar\, b_i \right)\ominus \left(w_c\,\obar\, p\,\obar\, h_c\left(w_c\right) \right) \ominus \left(c_i\,\obar\, h_r\left(w_r\right)\,\obar\, d_i \right).
	\]
	
\par\noindent	The language $DW_k$ is the smallest set including the empty picture and closed under nesting accretion and Simplot closure (see \eqref{eqDefSimplot} in Section~\ref{s-preliminaries}).
\end{definition}
Figure~\ref{figChinese+Accretion+Wellnested} (right) 
 illustrates accretion and (left) shows a picture in $DW_1$; for comparison a Chinese box picture of the same size is shown in the middle.

\begin{figure}[htb]
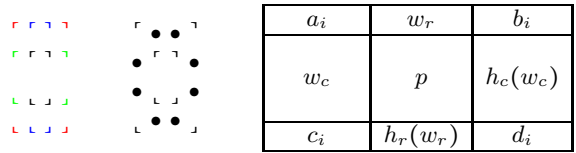

\begin{center}
$
\begin{array}{c c c c}
\red \ulcorner & \blue \ulcorner & \blue \urcorner& \red\urcorner 
\\
\green \ulcorner & \ulcorner & \urcorner &\green\urcorner 
\\
\green \llcorner & \llcorner &  \lrcorner & \green\lrcorner 
\\
\red \llcorner &\blue \llcorner & \blue \lrcorner & \red \lrcorner 
\end{array}
$
\qquad
$
\begin{array}{c c c c}
\ulcorner & \bullet & \bullet& \urcorner 
\\
\bullet & \ulcorner & \urcorner &\bullet
\\
\bullet & \llcorner &  \lrcorner & \bullet
\\
 \llcorner &\bullet& \bullet & \lrcorner 
\end{array}
$
\qquad
\setlength{\arraycolsep}{5pt}\renewcommand{\arraystretch}{1.0}
$
\begin{array}{|ccc|ccc|c|}\hline
&a_i&& \multicolumn{3}{c|}{w_r}& b_i
\\ \cline{1-7}
&& & \multicolumn{3}{c|}{ \quad }&
\\
&w_c&     & \multicolumn{3}{c|}{ p }&	 h_c(w_c)	
\\
&&& \multicolumn{3}{c|}{ \quad }&
\\ \cline{1-7}
&c_i& & \multicolumn{3}{c|}{ h_r(w_r)}&d_i
\\\hline
\end{array}
$
\end{center}
\caption{(Left) An example of picture in $DW_1$ and (middle) the similar Chinese box version. (Right) Scheme of nesting accretion.}~\label{figChinese+Accretion+Wellnested}
\end{figure}
 The definition can be explained intuitively by considering two distinct occurrences of a quadruple of matching corners: 
{ the subpictures delimited by each quadruple (i.e., their bounding boxes) are either disjoint, or included one into the other; or they overlap and a third box exists that ``minimally'' bounds both boxes. The third case is illustrated in Figure~\ref {figChinese+Accretion+Wellnested}, left, by the overlapping blue and green boxes.} 
\par It is immediate to see that for any size $(2m, 2n)$, $m,n \ge 1$, there is a picture in $DW_k$.
	\begin{theorem}\label{thm-wn-are rec}
		The language $DW_k$ is not (tiling system) recognizable, for every $k \ge 1$.
	\end{theorem} 	
	\begin{proof}
	By contradiction, assume that $DW_k$ is recognizable. Without loss of generality, we consider only the case $k=1$.
	Consider the following picture $p$ in $DW_1$: $\begin{array}{cc} a & b\\ c&d
	\end{array}
	$.
	From closure properties of REC, the language 	$p^{\obar +}$ is recognizable, 
	hence also the language: 
	$$
	R=\left(a^{\obar +}\, \obar\, b^{\obar +}\right) \ominus \left((a\, \ominus \, c) \obar\, p^{\obar +}\, \obar\, (b \, \ominus\, d) \right) \ominus \left(c^{\obar +}\, \obar\, d^{\obar +}\right). 
	$$ 
	A picture in $R$ has $a^+ b^+$ in the top row and $c^+ d^+$ in the bottom row. 
	Let $T$ be the language obtained by intersection of $DW_1$  with $R^{\ominus +}$. Therefore, both $T$ and $T^{\ominus+}$ are also recognizable; moreover, the first row of every picture in $T^{\ominus+}$ has the form $a^nb^n$. 
	By applying the Horizontal Iteration Lemma of \cite{GiammRestivo1997} (Lemma 9.1) to $T^{\ominus+}$, there exists a (suitably large) picture $t$ in $T^{\ominus+}$ which can be written as the horizontal concatenation of the three (non empty) pictures $x,q, y$, namely $t=x \obar q \obar y$, such that $x \obar q^{i\obar} \obar y$ is also in $T^{\ominus+}$, a contradiction with the fact that the top row of the pictures in $T^{\ominus+}$ must be of the form $a^n b^n$.
	\end{proof}
\par

\subsection{Neutralizable Dyck language}
We investigate a possible definition of Dyck picture languages by means of a neutralization rule analogous to the congruence \eqref{eqNeutralization1DDyck} of Dyck word languages.

\begin{definition}[neutralizable Dyck language]\label{def2DDyckByNeutralization}
	Let $N$ be a new symbol not in $\Delta_k$. 
	The {\em neutralization relation } $\stackrel{\nu}\to \subseteq \left(\{N\}\cup\Delta_k \right)^{++}\times \left(\{N\}\cup\Delta_k \right)^{++}$, 
	is the smallest relation such that for every pair of pictures $p,p'$ in $\left(\{N\}\cup\Delta_k\right)^{++}$, $p \stackrel{\nu}\to p'$ if there are $m,n\ge 2$ and $1\le i \le k$, such that $p'$ is obtained from
	$p$ by replacing a subpicture of $p$ 
	of the form: 
	\begin{equation}\label{eq-neutralization}
		(a_i \ominus N^{m-2,1} \ominus c_i)\obar N^{m,n-2}\obar (b_i \ominus N^{m-2,1} \ominus d_i).
	\end{equation}
with the picture of the same size $N^{m,n}$.
	\par\noindent
	The 2D \emph{neutralizable Dyck language}, denoted with $DN_k\subseteq \Delta_k^{++}$, is the set of pictures $p$ such that there exists $p' \in N^{++}$ with $p \stackrel{\nu}{\to}^+ p'$.
\end{definition}
In other words, a $DN_k$  picture is transformed into a picture in  $N^{++}$ by a series of neutralizations.
It is obvious that the order of application of the neutralization steps is irrelevant for deciding if a picture is neutralizable. 

\begin{example}[neutralizations]\label{exNeutralizations}
	The following picture $p_1$ on the alphabet $\Delta_1$ is in $DN_1$  since it reduces to the neutral one by means of a sequence of six neutralization steps:
	\begin{center}
		\scalebox{0.95}{
		$
		p_1 =
		\begin{array}{c c c c c c}
		\ulcorner & \ulcorner & \urcorner & \ulcorner & \urcorner & \urcorner 
		\\
		 \ulcorner & \ulcorner & \urcorner & \llcorner &\lrcorner & \urcorner 
		\\
		 \llcorner & \llcorner & \lrcorner & \ulcorner &\urcorner & \lrcorner 
		\\
		\llcorner & \llcorner & \lrcorner & \llcorner &\lrcorner & \lrcorner 
		\end{array}
		\stackrel{\nu}\to
		\begin{array}{c c c c c c}
		\ulcorner & \ulcorner & \urcorner & \ulcorner &\urcorner & \urcorner 
		\\
		\ulcorner & N & N & \llcorner &\lrcorner & \urcorner 
		\\
		\llcorner & N & N & \ulcorner &\urcorner & \lrcorner 
		\\
		\llcorner & \llcorner & \lrcorner & \llcorner &\lrcorner & \lrcorner 
		\end{array}
		\stackrel{\nu}\to
		\begin{array}{c c c c c c}
		\ulcorner & N & N & \ulcorner &\urcorner & \urcorner 
		\\
		\ulcorner & N & N & \llcorner &\lrcorner & \urcorner 
		\\
		\llcorner & N & N & \ulcorner &\urcorner & \lrcorner 
		\\
		\llcorner & N & N & \llcorner &\lrcorner & \lrcorner 
		\end{array}
				\stackrel{\nu}\to
		\begin{array}{c c c c c c}
		\ulcorner & N & N & N &N & \urcorner
		\\
		\ulcorner & N & N & N &N & \urcorner 
		\\
		\llcorner & N & N & \ulcorner &\urcorner & \lrcorner 
		\\
		\llcorner & N & N & \llcorner &\lrcorner & \lrcorner 
		\end{array}
				$
		}
		\\\medskip
		\scalebox{0.95}{
		$
		\stackrel{\nu}\to
		\begin{array}{c c c c c c}
		\ulcorner & N & N & N &N & \urcorner
		\\
		\ulcorner & N & N & N &N & \urcorner 
		\\
		\llcorner & N & N & N &N & \lrcorner 
		\\
		\llcorner & N & N & N &N & \lrcorner
		\end{array}
		\stackrel{\nu}\to
		\begin{array}{c c c c c c}
		\ulcorner & N & N & N &N & \urcorner
		\\
		N & N & N & N &N & N 
		\\
		N & N & N & N &N & N 
		\\
		\llcorner & N & N & N &N & \lrcorner
		\end{array}
		\stackrel{\nu}\to
		\begin{array}{c c c c c c}
		N & N & N & N &N & N 
		\\
		N & N & N & N &N & N 
		\\
		N & N & N & N &N & N 
		\\
		N & N & N & N &N & N 
		\end{array}
		$
		}
	\end{center}
	Neutralizations have been arbitrarily applied in top to bottom, left to right order. 
\end{example}
\par
 By a proof almost identical to the one of Theorem~\ref{thm-wn-are rec}, since the language $T^{\ominus+}$ can be obtained from $DN_k$  by intersection
with a recognizable language, we have:
	\begin{theorem}\label{thm-neut-arenot-rec}
		The language $DN_k$  is not (tiling system) recognizable for every $k\ge 1$.
	\end{theorem}
 
\par
Although $DW_k$ is defined by a diverse mechanism, the next inclusion is immediate. 
\begin{theorem}
	The language $DW_k$  is strictly included in $DN_k$ for every $k\ge 1$.
\end{theorem}
\begin{proof}
The inclusion $DW_k \subseteq DN_k$ is obvious since any picture in $DW_k$  can be neutralized in accordance with Definition~\ref{def2DDyckByNeutralization}.
Then the thesis follows since the neutralizable picture 
$
p_N =
\begin{array}{c c c c}
\ulcorner & \ulcorner & \urcorner & \urcorner 
\\
\llcorner & \llcorner & \lrcorner &\lrcorner 
\\
\end{array}
$
\; cannot be obtained using nesting accretion. 
\end{proof}
 Another picture in $DN_1 \setminus DW_1$ is in Figure~\ref{fig-rectangles}.

\section{Row-column combination of Dyck languages}\label{sectDyckCrosswords}
We consider the pictures such that their rows and columns are Dyck languages, more precisely, they are Dyck word languages over the same alphabet but with different pairing of terminal characters. Such pictures, called Dyck crosswords, may be viewed as analogous of Dyck word languages.

\par
Following \cite{GiammRestivo1997} we introduce the row-column combination operation that takes two word languages and produces a picture language.
\begin{definition}[row-column combination a.k.a. crossword]\label{defRowColCombination}
	Let $S' , S'' \subseteq \Sigma^*$ be two word languages, called \emph{component languages}. The \emph{row-column combination} or \emph{crossword} of $S'$ and $S''$ is the picture language $L$
	such that a picture $p \in \Sigma^{++}$ belongs to $L$ if, and only if, the words corresponding to each row (in left-to-right order) and to each column (in top-down order) of $p$ belong to $S'$ and $S''$, respectively. 
	\par\noindent
\end{definition}

{The row-column combination of regular languages has received attention in the past since its alphabetic projection exactly coincide with the REC family~\cite{GiammRestivo1997}; some complexity issues for this case are addressed in the recent paper~\cite{DBLP:journals/iandc/FennerPT22} where the combinations are called ``regex crosswords''. 
Moreover, given two regular languages $S',S''$, it is undecidable to establish whether their composition is empty.}
In this section, we investigate the properties of the row-column combination of a fundamental type of context-free languages, the Dyck ones.

\par
 The picture alphabet is the same of $DW_k$  and $DN_k$  languages, here preferably represented by letters instead of corner symbols.

\begin{definition}[Dyck crossword alphabet and language]\label{defCrosswordAlphabet}
Let $\Delta_k=\{a_i, b_i, c_i, d_i \mid 1 \le i \le k\}$, an alphabet. We associate $\Delta_k$ with two different Dyck alphabets, the \emph{Dyck row alphabet} 
	$\Delta^{Row}_k$ and the \emph{Dyck column} \emph{alphabet} $\Delta^{Col}_k$ by means of the following matching pairs:
\[
\left\{
\begin{array}{ll}
\text{for }	\Delta^{Row}_k : & \left\{[ a_i, b_i] \mid i \le 1 \le k \right\} \cup\left \{ [c_i,d_i]\mid 1 \le i \le k \right\}
\\
\text{for }	\Delta^{Col}_k : & \left \{[ a_i, c_i] \mid i \le 1 \le k \right\} \cup \left\{ [b_i,d_i]\mid 1 \le i \le k \right\}
\end{array}
\right. .
\]

%
The corresponding Dyck languages, without $\varepsilon$, are denoted by $D^{Row}_k \subset {\Delta_k}^+$ and $D^{Col}_k \subset {\Delta_k}^+$.

	\par
	The \emph{Dyck crossword} language $DC_k$ is the row-column combination of $D^{Row}_k$ and $D^{Col}_k$. 
\end{definition}

In the following, we often consider only the language $DC_1$, over alphabet $\{a,b,c,d\}$, when statements and properties of $DC_k$ are straighforward generalizations of the $DC_1$ case. 
\begin{remark}
The choice in Definition~\ref{defCrosswordAlphabet} that the $DC_k$ alphabet $\Delta_k$ consists of one or more quadruples $a_i,b_i,c_i,d_i$, $1 \le i \le k$, is not just for continuity with the alphabet of the well-nested and neutralizable cases, but it is imposed by the following simple facts. For brevity we consider $k=1$.
\begin{enumerate}[(i)]
\item Let $\Gamma$ be the binary alphabet $\{e, e'\}$. Let $S'$ and $S''$ be the Dyck languages respectively for rows and for columns based on  $\Gamma$, with matching parentheses $(e,e')$ for both rows and columns. 
Then, it is easy to see that the row-column combination of $S'$ and $S''$ is empty, since it is impossible to complete a $DC$ picture starting from a row containing word $e e'$.
 Moreover, the combination remains empty if we invert the matching for columns to $(e', e)$.
	
\item Let the alphabet for words be  $\Gamma= \{e, e', f, f'\}$. Then, to obtain a non-empty combination, there is only one way (disregarding trivially equivalent letter permutations) of matching the letters, namely: for rows, $(e,e'), (f, f')$ and for columns $(e,f), (e', f')$.
For instance, the choice $(e,f'), (e', f)\}$ for columns does not produce any $DC_1$ picture.
By renaming the letters of $\Gamma$ as $\Delta_1= \{a, b, c, d \}$ we regain the row/column Dyck alphabets of Definition~\ref{defCrosswordAlphabet}; then, the matching $\Delta^{Row}_1=\{a,b\}\cup \{c, d \}$ and
$\Delta^{Col}_1=\{a, d\}\cup \{b, c\}$ makes $DC_1$ empty.

\item Let the alphabet for words have six letters $\Gamma= \{e, e', f, f', g, g'\}$. From part (i) it is easy to see that, no matter what matching is chosen for row and columns, two of the letters cannot occur in any picture of $DC_1$. Therefore, it is enough to consider an alphabet of size multiple of four.
\item A consequence of the previous items is that the following property of Dyck words over a binary alphabet $\{e, e'\}$ does not hold for $DC_1$: any Dyck word, e.g., $e' e$, occurs as a factor of some Dyck word, e.g., $e\, e' e\, e'$; this is not true for the rows and the columns of Dyck crosswords because each one of the row/column Dyck alphabets contains two pairs of symbols, not just one. For instance the word $ad$ is a forbidden factor of language $D^{Row}_1$. 
\end{enumerate}
\end{remark}

\par
We state and prove some basic properties.
It is easy to notice that $DN_k\subseteq DC_k$: for instance, when neutralizing a subpicture, the neutralization of its two corners $(a_i,b_i)$ acts in that row as the neutralization rule for words in $D_k^{row}$, and similarly for the other corners. We later prove that this inclusion is proper. 
\par
The result of not being tiling recognizable holds also for $DC_k$:
\begin{theorem}\label{theorRECincomparableDC}
	For every $k\ge 1$, the language $DC_k$ is not (tiling system) recognizable.
\end{theorem}
\begin{proof}
The proof is essentially the same as of Theorem~\ref{thm-wn-are rec}, since also in this case the language $T^{\ominus+}$ can be obtained from $DC_1$ by intersection
with a recognizable language.
\end{proof}

{ The next property of $DC_k$ is that any picture $p$ that is partitioned into  $DC_k$ subpictures is also in $DC_k$. This is obvious since each row of $p$ is the concatenation of Dyck words, and similarly for columns. An analogous result holds for each language $DN_k$ (for $DW_k$ this holds by definition). }
	\begin{theorem}[Invariance under Simplot operation]\label{theorClosureDCunderSimplot}
		$(DC_k)^{++} = DC_k$ and $(DN_k)^{++} = DN_k$.
	\end{theorem}

\par                                               
{ Another question for any of the Dyck-like picture languages introduced is whether its row and column languages respectively saturate the horizontal and vertical Dyck word languages. We prove that this is the case for $DN_k$ and $DC_k$, but this is not for $DW_k$.
	{ Let $\Delta_k=\{a_i,b_i,c_i,d_i\mid 1\le i \le k\}$. Let $P\subseteq \Delta^{++}_k$ be a picture language and define the \emph{row language} of $P$ as: 
		$\text{ROW}(P)= \{w \in \Delta^+_k \mid \text{ there exist } p\in P, p',p''\in \Delta_k^{++} \text{ such that } p=p'\ominus w \ominus p''\}$. The column language of $P$, $\text{COL}(P)$ is defined analogously.}
	                                                
	\begin{theorem}[row/column languages]~\label{theorRow/ColumnLanguages}             
		\begin{enumerate}                                       
			\item $\text{ROW}(DC_k)=\text{ROW}(DN_k) = D^{Row}_k$, $\text{COL}(DC_k)=\text{COL}(DN_k)= D^{Col}_k$.
			\item $\text{ROW}(DW_k) \subsetneq D^{Row}_k$, $\text{COL}(DW_k) \subsetneq D^{Col}_k$.    
		\end{enumerate}                                        
                                                 
	\end{theorem}                                          
	\begin{proof}                                          
	Part (1): It is enough to prove that $D^{Row}_k\subseteq \text{ROW}(DN_k)$, since the other inclusion is obvious and the case for columns is symmetrical; moreover, $DN_k \subseteq DC_k$, so there is no need to prove the statement for $DC_k$. Without loss of generality, we consider only the case $k=1$. 
		We prove by induction on $n\ge 2 $, that for every word $w \in D^{Row}_1$ of length $n$ there exists a picture $p \in DN_1$ of the form $w_1 \ominus w_2 \ominus w \ominus w_3$ for $w_1,w_2,w_3 \in D^{Row}_1$. There are two base cases, the words $ab$ and $cd$. 
		The word $ab$ is (also) the third row in the $DN_1$ picture                  
		$ab\ominus cd\ominus ab \ominus cd$,                              
		while $cd$ is (also) the third row in the $DN_1$ picture                    
		$ab\ominus ab\ominus cd \ominus cd$.                              
		The induction step has three cases: a word $w\in D^{Row}_1$ of length $n>2$ has the form $w'w''$, or the form $a w' b$ or the form $c w' d$, for some $w',w'' \in D^{Row}_1$ of length less than $n$. Let $p',p''$ the pictures verifying the induction hypothesis for $w'$ and $w''$ respectively.
		The case of concatenation $w'w''$ is obvious (just consider the picture $p'\obar p''$). The case $a w' b$ can be solved by considering the picture $(a \ominus c \ominus a \ominus c)\obar p' \obar (b \ominus d \ominus b \ominus d)$, which is in $DN_1$. Similarly, for the case $c w' d$ just consider the $DN_1$ picture $(a \ominus a \ominus c \ominus c)\obar p' \obar (b \ominus b \ominus d \ominus d)$.
\par\noindent                                          
	Part (2): The Dyck word $abcd$ cannot be a row of a picture in $DW_k$. In fact, every picture in $DW_1$ of width 4 must be in the vertical concatenation closure of the set composed of the following two pictures, which do not include an $abcd$ row :
\scalebox{0.8}{$\begin{array}{c c c c}                              
 a &  a &  b&  b \\                                     
 a &  a &  b & b \\                                      
 c & c &  d &  d \\                                     
 c &  c &  d &  d \\                                    
\end{array}, \qquad                                       
\begin{array}{cc} a & b \\ c& d \end{array}\begin{array}{cc} a & b \\ c& d\end{array}.   
$                                                
}                                                
\end{proof}

\subsection{Matching-graph circuits}                               
                                                 
\begin{figure}
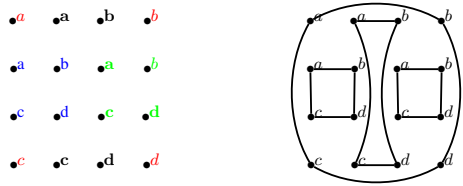
                                          
	\setlength{\tabcolsep}{30pt}                                  
	\setlength{\arraycolsep}{7pt}\renewcommand{\arraystretch}{1.2}                 
	\begin{center}                                         
		                                                
		\psset{arrows=-,labelsep=5pt,colsep=15pt,rowsep=15pt,nodealign=true}              
		\scalebox{0.7}{\begin{psmatrix}                                      
				\dotnode{p11}{$\red a$}&\dotnode{p12}{$\mathbf a$}&\dotnode{p13}{$\mathbf b$}&\dotnode{p14}{$\red b$}      
				\\                                              
				\dotnode{p21}{$\blue\text a$}&\dotnode{p22}{$\blue\text b$}&\dotnode{p23}{$\green\mathbf{a} $}&\dotnode{p24}{$\green b$}      
				\\                                              
				\dotnode{p31}{$\blue \text c$}&\dotnode{p32}{$\blue \text d$}&\dotnode{p33}{$\green \mathbf{c}$}&\dotnode{p34}{$\green \mathbf{d}$}      
				\\                                           
				\dotnode{p41}{$\red c$}&\dotnode{p42}{$\mathbf c$}&\dotnode{p43}{$\mathbf d$}&\dotnode{p44}{$\red d$}    
			\end{psmatrix}
		}\qquad\qquad\qquad
		\psset{arrows=-,labelsep=5pt,colsep=15pt,rowsep=15pt,nodealign=true}         
		\scalebox{0.7}{\begin{psmatrix}                                      
				\dotnode{p11}{$a$}&\dotnode{p12}{$a$}&\dotnode{p13}{$b$}&\dotnode{p14}{$b$}      
				\\                                           
				\dotnode{p21}{$a$}&\dotnode{p22}{$b$}&\dotnode{p23}{$a$}&\dotnode{p24}{$b$}      
				\\                                           
				\dotnode{p31}{$c$}&\dotnode{p32}{$d$}&\dotnode{p33}{$c$}&\dotnode{p34}{$d$}      
				\\                                           
				\dotnode{p41}{$c$}&\dotnode{p42}{$c$}&\dotnode{p43}{$d$}&\dotnode{p44}{$d$}                                       
				
				\ncarc[arcangle=30,linestyle=solid,linewidth=1.0pt]{p11}{p14} 
				\ncarc[arcangle=0,linestyle=solid,linewidth=1.0pt]{p12}{p13}             
				\ncarc[arcangle=-30, linestyle=solid,linewidth=1.0pt]{p11}{p41}
				\ncarc[arcangle=30,linestyle=solid,linewidth=1.0pt]{p14}{p44}
				\ncarc[arcangle=25,linestyle=solid,linewidth=1.0pt]{p12}{p42}
				\ncarc[arcangle=-25,linestyle=solid,linewidth=1.0pt]{p13}{p43}
				
				\ncarc[arcangle=0,linestyle=solid,linewidth=1.0pt]{p21}{p22}
				\ncarc[arcangle=0,linestyle=solid,linewidth=1.0pt]{p23}{p24}
				\ncarc[arcangle=0,linestyle=solid,linewidth=1.0pt]{p21}{p31}
				\ncarc[arcangle=0,linestyle=solid,linewidth=1.0pt]{p22}{p32}
				\ncarc[arcangle=0,linestyle=solid,linewidth=1.0pt]{p23}{p33} 
				\ncarc[arcangle=0,linestyle=solid,linewidth=1.0pt]{p24}{p34} 
				
				\ncarc[arcangle=0,linestyle=solid,linewidth=1.0pt]{p31}{p32}
				\ncarc[arcangle=0,linestyle=solid,linewidth=1.0pt]{p33}{p34}
				
				\ncarc[arcangle=0,linestyle=solid,linewidth=1.0pt]{p42}{p43}           
				\ncarc[arcangle=-30,linestyle=solid,linewidth=1.0pt]{p41}{p44} 
			\end{psmatrix}
		}
	\end{center}
	\caption{(Left) A $DC_1$ picture whose cells are partitioned into 4 quadruples of matching symbols, identified by the same color (font). (Right) An alternative visualization (right) by a graph using edges that connect matching symbols. }
	\label{fig-rectangles}
\end{figure}
\par
We present some patterns that occur in $DC_k$ pictures. The simplest patterns are found in pictures that are partitioned into rectangular circuits connecting four elements, see, e.g., Figure~\ref{fig-rectangles}, right, where an edge connects two symbols on the same row (or column) which match in the row (column) Dyck word.
Notice that the graph made by the edges contains four disjoint circuits of length four, called \emph{rectangles} for brevity. Three of the  circuits are nested inside the outermost one.
\par
However, a picture in $DC_1$ may also include circuits longer than four. In Figure~\ref{figPictureFarfallaInDC} (left) we see a circuit of length 12, labeled by the word $(abdc)^3$, and on the right a circuit of length 36. Notice that when a picture on $\Delta_1$ is represented by circuits, the node labels are redundant since they are uniquely determined on each circuit.

\begin{figure}[hbt]
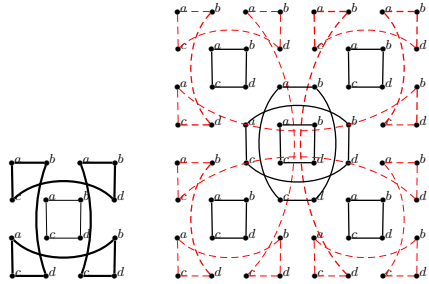

	\begin{center}
		\scalebox{0.55}{%
			\psset{arrows=-,labelsep=3pt,colsep=15pt,rowsep=15pt,nodealign=true}         
			\begin{psmatrix}                                                                         
				\dotnode{p11}{$a$}&\dotnode{p12}{$b$}&\dotnode{p13}{$a$}&\dotnode{p14}{$b$}      
				\\                                           
				\dotnode{p21}{$c$}&\dotnode{p22}{$a$}&\dotnode{p23}{$b$}&\dotnode{p24}{$d$}       
				\\                                           
				\dotnode{p31}{$a$}&\dotnode{p32}{$c$}&\dotnode{p33}{$d$}&\dotnode{p34}{$b$}       
				\\                                           
				\dotnode{p41}{$c$}&\dotnode{p42}{$d$}&\dotnode{p43}{$c$}&\dotnode{p44}{$d$}                                          
				\ncarc[arcangle=0, linestyle=solid,linewidth=1.5pt]{p11}{p12}              
				\ncarc[arcangle=-0, linestyle=solid,linewidth=1.5pt]{p11}{p21}            
				\ncarc[arcangle=-20, linestyle=solid,linewidth=1.5pt]{p12}{p42}           
				\ncarc[arcangle=-0, linestyle=solid,linewidth=1.5pt]{p41}{p42}                                           
				\ncarc[arcangle=45, linestyle=solid,linewidth=1.5pt]{p21}{p24}             
				\ncarc[arcangle=-45, linestyle=solid,linewidth=1.5pt]{p31}{p34}             
				\ncarc[arcangle=-0, linestyle=solid,linewidth=1.5pt]{p31}{p41}             
				\ncarc[arcangle=-0, linestyle=solid,linewidth=1.5pt]{p34}{p44}                                            
				\ncarc[arcangle=0, linestyle=solid,linewidth=1.5pt]{p13}{p14}            
				\ncarc[arcangle=-0, linestyle=solid,linewidth=1.5pt]{p43}{p44}            
				\ncarc[arcangle=20, linestyle=solid,linewidth=1.5pt]{p13}{p43}           
				\ncarc[arcangle=0, linestyle=solid,linewidth=1.5pt]{p14}{p24}                                           
				\ncarc[arcangle=0, linestyle=solid,linewidth=0.5pt]{p22}{p23}             
				\ncarc[arcangle=0, linestyle=solid,linewidth=0.5pt]{p32}{p33}             
				\ncarc[arcangle=0, linestyle=solid,linewidth=0.5pt]{p22}{p32}              
				\ncarc[arcangle=-0, linestyle=solid,linewidth=0.5pt]{p23}{p33}             
			\end{psmatrix}	
		}%
		\qquad
		\scalebox{0.55}{%
\psset{arrows=-,labelsep=3pt,colsep=15pt,rowsep=15pt,nodealign=true}         
\begin{psmatrix}                                 
\\    
\dotnode{p11}{$a$}&\dotnode{p12}{$b$}&\dotnode{p13}{$a$}&\dotnode{p14}{$b$}
& \dotnode{p15}{$a$}&\dotnode{p16}{$b$}&\dotnode{p17}{$a$}&\dotnode{p18}{$b$}    
\\   
\dotnode{p21}{$c$}&\dotnode{p22}{$a$}&\dotnode{p23}{$b$}&\dotnode{p24}{$d$}
& \dotnode{p25}{$c$}&\dotnode{p26}{$a$}&\dotnode{p27}{$b$}&\dotnode{p28}{$d$}            
\\   
\dotnode{p31}{$a$}&\dotnode{p32}{$c$}&\dotnode{p33}{$d$}&\dotnode{p34}{$a$}
&  \dotnode{p35}{$b$}&\dotnode{p36}{$c$}&\dotnode{p37}{$d$}&\dotnode{p38}{$b$}          
\\                                             
\dotnode{p41}{$c$}&\dotnode{p42}{$d$}&\dotnode{p43}{$a$}&\dotnode{p44}{$a$} 
&    \dotnode{p45}{$b$}&\dotnode{p46}{$b$}&\dotnode{p47}{$c$}&\dotnode{p48}{$d$}
\\                                             
\dotnode{p51}{$a$}&\dotnode{p52}{$b$}&\dotnode{p53}{$c$}&\dotnode{p54}{$c$} 
&    \dotnode{p55}{$d$}&\dotnode{p56}{$d$}&\dotnode{p57}{$a$}&\dotnode{p58}{$b$}
\\ 
\dotnode{p61}{$c$}&\dotnode{p62}{$a$}&\dotnode{p63}{$b$}&\dotnode{p64}{$c$} 
&    \dotnode{p65}{$d$}&\dotnode{p66}{$a$}&\dotnode{p67}{$b$}&\dotnode{p68}{$d$}
\\ 
\dotnode{p71}{$a$}&\dotnode{p72}{$c$}&\dotnode{p73}{$d$}&\dotnode{p74}{$b$} 
&    \dotnode{p75}{$a$}&\dotnode{p76}{$c$}&\dotnode{p77}{$d$}&\dotnode{p78}{$b$}
\\
\dotnode{p81}{$c$}&\dotnode{p82}{$d$}&\dotnode{p83}{$c$}&\dotnode{p84}{$d$} 
&    \dotnode{p85}{$c$}&\dotnode{p86}{$d$}&\dotnode{p87}{$c$}&\dotnode{p88}{$d$}

\ncarc[arcangle=0, linestyle=dashed,linecolor=red, linewidth=0.5pt]{p11}{p12}                 
\ncarc[arcangle=0, linestyle=dashed,linecolor=red,linewidth=0.5pt]{p11}{p21}                   
\ncarc[arcangle=-45, linestyle=dashed,linecolor=red,linewidth=0.9pt]{p12}{p42}

\ncarc[arcangle=0, linestyle=dashed,linecolor=red,linewidth=0.5pt]{p13}{p14}
\ncarc[arcangle=45, linestyle=dashed,linecolor=red,linewidth=0.5pt]{p13}{p83} 
                 
\ncarc[arcangle=-0, linestyle=dashed,linecolor=red,linewidth=0.5pt]{p14}{p24}
                   
\ncarc[arcangle=0, linestyle=dashed,linecolor=red,linewidth=0.9pt]{p15}{p16}
\ncarc[arcangle=0, linestyle=dashed,linecolor=red,linewidth=0.9pt]{p15}{p25}
\ncarc[arcangle=-45, linestyle=dashed,linecolor=red,linewidth=0.9pt]{p16}{p86}
																																																			
\ncarc[arcangle=0, linestyle=dashed,linecolor=red,linewidth=0.9pt]{p17}{p18}
\ncarc[arcangle=45, linestyle=dashed,linecolor=red,linewidth=0.9pt]{p17}{p47}
\ncarc[arcangle=0, linestyle=dashed,linecolor=red,linewidth=0.9pt]{p18}{p28} 

\ncarc[arcangle=45, linestyle=dashed,linecolor=red,linewidth=0.5pt]{p21}{p24}
\ncarc[arcangle=0, linestyle=solid,linewidth=0.9pt]{p22}{p23}
\ncarc[arcangle=0, linestyle=solid,linewidth=0.9pt]{p22}{p32}
\ncarc[arcangle=0, linestyle=solid,linewidth=0.9pt]{p23}{p33}
\ncarc[arcangle=45, linestyle=dashed,linecolor=red,linewidth=0.5pt]{p25}{p28}
\ncarc[arcangle=0, linestyle=solid,linewidth=0.9pt]{p26}{p27}
\ncarc[arcangle=0, linestyle=solid,linewidth=0.9pt]{p26}{p36}
\ncarc[arcangle=0, linestyle=solid,linewidth=0.9pt]{p27}{p37}

\ncarc[arcangle=-45, linestyle=dashed,linecolor=red,linewidth=0.5pt]{p31}{p38}
\ncarc[arcangle=0, linestyle=dashed,linecolor=red,linewidth=0.5pt]{p31}{p41}
\ncarc[arcangle=0, linestyle=solid,linewidth=0.9pt]{p32}{p33}
\ncarc[arcangle=0, linestyle=solid,linewidth=0.9pt]{p34}{p35}
\ncarc[arcangle=-45, linestyle=solid,linewidth=0.9pt]{p34}{p64}
\ncarc[arcangle=45, linestyle=solid,linewidth=0.9pt]{p35}{p65}
\ncarc[arcangle=0, linestyle=solid,linewidth=0.9pt]{p36}{p37}
\ncarc[arcangle=0, linestyle=dashed,linecolor=red,linewidth=0.9pt]{p38}{p48}

\ncarc[arcangle=0, linestyle=dashed,linecolor=red,linewidth=0.9pt]{p41}{p42}
\ncarc[arcangle=45, linestyle=solid,linewidth=0.9pt]{p43}{p46}
\ncarc[arcangle=0, linestyle=solid,linewidth=0.9pt]{p43}{p53}
\ncarc[arcangle=0, linestyle=solid,linewidth=0.9pt]{p44}{p45}
\ncarc[arcangle=0, linestyle=solid,linewidth=0.9pt]{p44}{p54}
\ncarc[arcangle=0, linestyle=solid,linewidth=0.9pt]{p45}{p55}
\ncarc[arcangle=0, linestyle=solid,linewidth=0.9pt]{p46}{p56}
\ncarc[arcangle=0, linestyle=dashed,linecolor=red,linewidth=0.5pt]{p47}{p48}

\ncarc[arcangle=0, linestyle=dashed,linecolor=red,linewidth=0.5pt]{p51}{p52}
\ncarc[arcangle=0, linestyle=dashed,linecolor=red,linewidth=0.5pt]{p51}{p61}
\ncarc[arcangle=-45, linestyle=dashed,linecolor=red,linewidth=0.9pt]{p52}{p82}
\ncarc[arcangle=-45, linestyle=solid,linewidth=0.9pt]{p53}{p56}
\ncarc[arcangle=0, linestyle=solid,linewidth=0.9pt]{p54}{p55}
\ncarc[arcangle=0, linestyle=dashed,linecolor=red,linewidth=0.5pt]{p57}{p58}
\ncarc[arcangle=45, linestyle=dashed,linecolor=red,linewidth=0.5pt]{p57}{p87}
\ncarc[arcangle=0, linestyle=dashed,linecolor=red,linewidth=0.5pt]{p58}{p68}

\ncarc[arcangle=45, linestyle=dashed,linecolor=red,linewidth=0.5pt]{p61}{p68}
\ncarc[arcangle=0, linestyle=solid,linewidth=0.9pt]{p62}{p63}
\ncarc[arcangle=0, linestyle=solid,linewidth=0.9pt]{p62}{p72}
\ncarc[arcangle=0, linestyle=solid,linewidth=0.9pt]{p63}{p73}
\ncarc[arcangle=0, linestyle=solid,linewidth=0.9pt]{p64}{p65}
\ncarc[arcangle=0, linestyle=solid,linewidth=0.9pt]{p66}{p67}
\ncarc[arcangle=0, linestyle=solid,linewidth=0.9pt]{p66}{p76}
\ncarc[arcangle=0, linestyle=solid,linewidth=0.9pt]{p67}{p77}

\ncarc[arcangle=-45, linestyle=dashed,linecolor=red,linewidth=0.5pt]{p71}{p74}
\ncarc[arcangle=0, linestyle=dashed,linecolor=red,linewidth=0.5pt]{p71}{p81}
\ncarc[arcangle=0, linestyle=solid,linewidth=0.9pt]{p72}{p73}
\ncarc[arcangle=0, linestyle=dashed,linecolor=red,linewidth=0.5pt]{p74}{p84}
\ncarc[arcangle=0, linestyle=dashed,linecolor=red,linewidth=0.5pt]{p75}{p85}
\ncarc[arcangle=-45, linestyle=dashed,linecolor=red,linewidth=0.5pt]{p75}{p78}
\ncarc[arcangle=0, linestyle=solid,linewidth=0.9pt]{p76}{p77}
\ncarc[arcangle=0, linestyle=dashed,linecolor=red,linewidth=0.5pt]{p78}{p88}

\ncarc[arcangle=0, linestyle=dashed,linecolor=red,linewidth=0.5pt]{p81}{p82}
\ncarc[arcangle=0, linestyle=dashed,linecolor=red,linewidth=0.5pt]{p83}{p84}
\ncarc[arcangle=0, linestyle=dashed,linecolor=red,linewidth=0.5pt]{p85}{p86}
\ncarc[arcangle=0, linestyle=dashed,linecolor=red,linewidth=0.5pt]{p87}{p88}
                
\end{psmatrix}
}%
	\end{center}
	\caption{Two pictures in $DC_1$.  (Left) The picture is partitioned into two circuits of length 12 and 4. (Right) The picture includes a circuit of length 36 and seven rectangular circuits. Its pattern embeds four partial copies (direct or rotated) of the left picture; in, say, the NW copy the ``triangle'' $b d c$ has been changed to $a a a$. Such a transformation can be reiterated to grow a series of pictures.}\label{figPictureFarfallaInDC}
\end{figure}

\par
We formally define the graph, situated on the picture grid, made by such circuits. 
	
\begin{definition}[matching graph]~\label{defPictureGraph}
	The \emph{matching graph} {associated} with a picture $p \in DC_k$, of size $(m,n)$, 
 is a pair $(V,E)$ where 
		the set $V$ of nodes is the set $\{1,\dots n \} \times \{1 \dots m\}$ and the set $E$ of edges is partitioned in two sets of {\emph row} and {\emph column} edges defined as follows, for all $1 \le i \le n, 1\le j \le m$: 
	\begin{itemize}
		\item for all pairs of matching letters $p_{i,j}, p_{i,j'}$ in $\Delta^{Row}_k$, with $j< j'\le m$, there is a row (horizontal) edge connecting $(i,j)$ with $(i,j')$, 
		\item for all pairs of matching letters $p_{i,j}, p_{i',j}$ in $\Delta^{Col}_k$, with $i< i'\le n$, there is a column (vertical) edge connecting $(i,j)$ with $(i',j)$,
		
	\end{itemize}
\end{definition}

Therefore, there is a horizontal edge connecting two matching letters $a_i,b_i$ or $c_i,d_i$ that occur in the same row: e.g., the edge $(2,1)\leftrightarrow (2,4)$ of Figure~\ref{figPictureFarfallaInDC}, left. Analogously, there is a vertical edge connecting two matching letters $a_i,c_i$ or $b_i,d_i$, that occur in the same column: e.g., the edge $(2,2)\leftrightarrow (3,2)$ of Figure~\ref{figPictureFarfallaInDC}, left.

{ From elementary properties of Dyck languages it follows that the distance on the picture grid 
between two nodes connected by an edge is an odd number.
}

	\begin{theorem}[matching-graph circuits]~\label{theorGraphCircuits}
		Let $p$ be a picture in $DC_k$. Then:
		\begin{enumerate}
			\item its matching graph $G$ is partitioned into disjoint simple circuits; 
			\item the clockwise visit of any such circuit, starting from one of its nodes with label $a_j$, yields a word in the language $(a_jb_jd_jc_j)^+$, for all $1\le j\le k$.
		\end{enumerate}
	\end{theorem}

\begin{proof}
 Part (1): By Definition~\ref{defPictureGraph}, every node of $G$ has degree 2, with one row edge and one column edge, since its corresponding row and column in picture $p$ are Dyck words. Every node must be on a circuit, otherwise there would be a node of degree 1. Each circuit must be simple and the sets of nodes on two circuits are disjoint, else one of the nodes would have degree greater than 2.
Part (2) is obvious, since from a node labeled $a_j$ there is a row edge connecting with a node labeled $b_j$, for which there is a column edge connecting with a $d_j$, then a row edge 
connecting $d_j$ with $c_j$, etc., finally closing the circuit with a column edge connecting a $c_j$ with the original $a_j$. 
\end{proof}

Theorem~\ref{theorGraphCircuits} has a simple interpretation: to check that in a picture all rows and columns are Dyck words of respectively $D^{Row}_k$ and $D^{Col}_k$, we could proceed along each linear path. The process moves from an opening letter  (say $a$) to its matching letter ($b$) on the same row, while verifying that the word between the two letters is correctly parenthesized; then, the process moves to the closed matching letter ($d$) on the column of $b$, and so on, until the circuit is closed, or interrupted causing rejection.
Such a checking of $DC_k$ membership corresponds to a way of checking Dyck membership for words. Since a word is a picture of size $(1,n)$, its associated matching graph is the well-known so-called rainbow representation, e.g., \quad
\psset{arrows=-,labelsep=15pt,colsep=10pt,rowsep=15pt,nodealign=true}         
		\scalebox{0.9}{
		\begin{psmatrix}                                   
\pnode{p11}{$a$}&\pnode{p12}{$c$}&\pnode{p13}{$d$}&\pnode{p14}{$a$}&\pnode{p15}{$b$}&\pnode{p16}{$b$}                          
		\ncarc[nodesep=8pt,arcangle=40,linestyle=solid,linewidth=0.6pt]{p11}{p16} 
		\ncarc[nodesep=6pt,arcangle=80,linestyle=solid,linewidth=0.6pt]{p12}{p13}             
		\ncarc[nodesep=6pt,arcangle=80, linestyle=solid,linewidth=0.6pt]{p14}{p15}
		\end{psmatrix}
		\quad
}
of the syntax tree of the word. A matching circuit then corresponds to the binary relation between the two ends of a rainbow arc. However it is perhaps unexpected that moving from 1D to 2D the length of circular paths increases not just to $2\times 2$, but without an upper bound, as proved below.

Notice that there exist pictures that are not in $DC_1$, but which still can be partitioned in circuits with label in $abdc^+$ and having arcs following the correct directions (starting from a node label $a$, going  right, then down, then left and then up).  
For instance, in the picture: 
\[
\begin{array}{|cccccccc|}\hline
a&a&a&a&b&b&b&b\\
c&c&a&a&d&d&b&b\\
a&a&c&c&b&b&d&d\\
c&c&c&c&d&d&d&d\\	\hline
\end{array}
\]

\noindent all 8 columns and the first and fourth rows are Dyck words, while the second and third rows are not Dyck words. Still, it is easy to verify that the picture can be partitioned in ''correct`` circuits having label in $(abcd)^+$ (two circuits of length 12 and two circuits of length 4). 

\begin{figure}[htb]
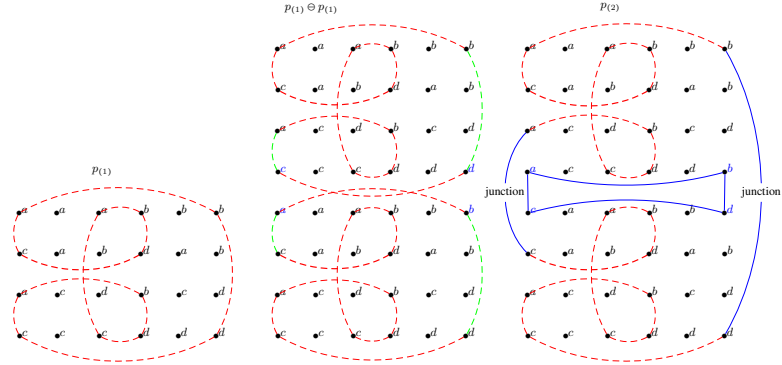

	\begin{center}
		\scalebox{0.5}{%
			\psset{arrows=-,labelsep=3pt,colsep=20pt,rowsep=20pt,nodealign=true} 
			\begin{psmatrix} 
				&&$p_{(1)}$ 
				\\    
				\dotnode{p11}{$a$}&\dotnode{p12}{$a $}&\dotnode{p13}{$a$}&\dotnode{p14}{$b $}&\dotnode{p15}{$b$}&\dotnode{p16}{$b $}
				\\   
				\dotnode{p21}{$c $}&\dotnode{p22}{$a $}&\dotnode{p23}{$b $}&\dotnode{p24}{$d$}&\dotnode{p25}{$a$}&\dotnode{p26}{$b $}
				\\   
				\dotnode{p31}{$a $}&\dotnode{p32}{$c $}&\dotnode{p33}{$d $}&\dotnode{p34}{$b $}&\dotnode{p35}{$c$}&\dotnode{p36}{$d $}
				\\
				\dotnode{p41}{$c$}&\dotnode{p42}{$c$}&\dotnode{p43}{$c$}&\dotnode{p44}{$d$}&\dotnode{p45}{$d$}&\dotnode{p46}{$d$} 
				\ncarc[arcangle=30, linestyle=dashed,linecolor=red,linewidth=0.5pt]{p11}{p16}                  
				\ncarc[arcangle=-30, linestyle=dashed,linecolor=red,linewidth=0.5pt]{p11}{p21}                  
				\ncarc[arcangle=30, linestyle=dashed,linecolor=red,linewidth=0.5pt]{p13}{p14}
				\ncarc[arcangle=-30, linestyle=dashed,linecolor=red,linewidth=0.5pt]{p13}{p43}
				\ncarc[arcangle=30, linestyle=dashed,linecolor=red,linewidth=0.5pt]{p14}{p24}
				\ncarc[arcangle=30, linestyle=dashed,linecolor=red,linewidth=0.5pt]{p16}{p46}
				
				\ncarc[arcangle=-30, linestyle=dashed,linecolor=red,linewidth=0.9pt]{p21}{p24}
				
				\ncarc[arcangle=30, linestyle=dashed,linecolor=red,linewidth=0.5pt]{p31}{p34}
				\ncarc[arcangle=-30, linestyle=dashed,linecolor=red,linewidth=0.5pt]{p31}{p41}
				\ncarc[arcangle=30, linestyle=dashed,linecolor=red,linewidth=0.5pt]{p34}{p44} 
				
				\ncarc[arcangle=-30, linestyle=dashed,linecolor=red,linewidth=0.5pt]{p41}{p46}
				\ncarc[arcangle=-30, linestyle=dashed,linecolor=red,linewidth=0.5pt]{p43}{p44}                   
			\end{psmatrix}
		}%
		\qquad
		\scalebox{0.5}{%
			\psset{arrows=-,labelsep=3pt,colsep=20pt,rowsep=20pt,nodealign=true} 
			\begin{psmatrix}  
				\pnode{p01}                                
				\\    
				\dotnode{p11}{$a$}&\dotnode{p12}{$a $}&\dotnode{p13}{$a$}&\dotnode{p14}{$b $}&\dotnode{p15}{$b$}&\dotnode{p16}{$b $}
				\\   
				\dotnode{p21}{$c $}&\dotnode{p22}{$a $}&\dotnode{p23}{$b $}&\dotnode{p24}{$d$}&\dotnode{p25}{$a$}&\dotnode{p26}{$b $}
				\\   
				\dotnode{p31}{$a $}&\dotnode{p32}{$c $}&\dotnode{p33}{$d $}&\dotnode{p34}{$b $}&\dotnode{p35}{$c$}&\dotnode{p36}{$d $}
				\\
				\dotnode{p41}{$\blue c$}&\dotnode{p42}{$c$}&\dotnode{p43}{$c$}&\dotnode{p44}{$d$}&\dotnode{p45}{$d$}&\dotnode{p46}{$\blue d$} 
				
				\nput{0}{p01}{$p_{(1)}\ominus p_{(1)}$ }
				
				\ncarc[arcangle=30, linestyle=dashed,linecolor=red,linewidth=0.5pt]{p11}{p16}
				\ncarc[arcangle=-30, linestyle=dashed,linecolor=red,linewidth=0.5pt]{p11}{p21}                  
				\ncarc[arcangle=30, linestyle=dashed,linecolor=red,linewidth=0.5pt]{p13}{p14}
				\ncarc[arcangle=-30, linestyle=dashed,linecolor=red,linewidth=0.5pt]{p13}{p43}
				\ncarc[arcangle=30, linestyle=dashed,linecolor=red,linewidth=0.5pt]{p14}{p24}
				\ncarc[arcangle=30, linestyle=dashed,linecolor=green,linewidth=0.5pt]{p16}{p46}
				
				\ncarc[arcangle=-30, linestyle=dashed,linecolor=red,linewidth=0.9pt]{p21}{p24}
				
				\ncarc[arcangle=30, linestyle=dashed,linecolor=red,linewidth=0.5pt]{p31}{p34}
				\ncarc[arcangle=-30, linestyle=dashed,linecolor=green,linewidth=0.5pt]{p31}{p41}
				\ncarc[arcangle=30, linestyle=dashed,linecolor=red,linewidth=0.5pt]{p34}{p44} 
				
				\ncarc[arcangle=-30, linestyle=dashed,linecolor=red,linewidth=0.5pt]{p41}{p46}
				\ncarc[arcangle=-30, linestyle=dashed,linecolor=red,linewidth=0.5pt]{p43}{p44}    
				\\    
				\dotnode{p'11}{$\blue a$}&\dotnode{p'12}{$a $}&\dotnode{p'13}{$a$}&\dotnode{p'14}{$b $}&\dotnode{p'15}{$b$}&\dotnode{p'16}{$\blue b $}
				\\   
				\dotnode{p'21}{$c $}&\dotnode{p'22}{$a $}&\dotnode{p'23}{$b $}&\dotnode{p'24}{$d$}&\dotnode{p'25}{$a$}&\dotnode{p'26}{$b $}
				\\   
				\dotnode{p'31}{$a $}&\dotnode{p'32}{$c $}&\dotnode{p'33}{$d $}&\dotnode{p'34}{$b $}&\dotnode{p'35}{$c$}&\dotnode{p'36}{$d $}
				\\
				\dotnode{p'41}{$c$}&\dotnode{p'42}{$c$}&\dotnode{p'43}{$c$}&\dotnode{p'44}{$d$}&\dotnode{p'45}{$d$}&\dotnode{p'46}{$d$} 
				
				\ncarc[arcangle=30, linestyle=dashed,linecolor=red,linewidth=0.5pt]{p'11}{p'16}                  
				\ncarc[arcangle=-30, linestyle=dashed,linecolor=green,linewidth=0.5pt]{p'11}{p'21}
				
				\ncarc[arcangle=30, linestyle=dashed,linecolor=red,linewidth=0.5pt]{p'13}{p'14}
				\ncarc[arcangle=-30, linestyle=dashed,linecolor=red,linewidth=0.5pt]{p'13}{p'43}
				\ncarc[arcangle=30, linestyle=dashed,linecolor=red,linewidth=0.5pt]{p'14}{p'24}
				\ncarc[arcangle=30, linestyle=dashed,linecolor=green,linewidth=0.5pt]{p'16}{p'46}
				
				\ncarc[arcangle=-30, linestyle=dashed,linecolor=red,linewidth=0.9pt]{p'21}{p'24}
				
				\ncarc[arcangle=30, linestyle=dashed,linecolor=red,linewidth=0.5pt]{p'31}{p'34}
				\ncarc[arcangle=-30, linestyle=dashed,linecolor=red,linewidth=0.5pt]{p'31}{p'41}
				\ncarc[arcangle=30, linestyle=dashed,linecolor=red,linewidth=0.5pt]{p'34}{p'44} 
				
				\ncarc[arcangle=-30, linestyle=dashed,linecolor=red,linewidth=0.5pt]{p'41}{p'46}
				\ncarc[arcangle=-30, linestyle=dashed,linecolor=red,linewidth=0.5pt]{p'43}{p'44}
			\end{psmatrix}
		}%
		\qquad
		\scalebox{0.5}{%
			\psset{arrows=-,labelsep=3pt,colsep=20pt,rowsep=20pt,nodealign=true} 
			\begin{psmatrix} 
				&&$p_{(2)}$ 
				\\    
				\dotnode{p11}{$a$}&\dotnode{p12}{$a $}&\dotnode{p13}{$a$}&\dotnode{p14}{$b $}&\dotnode{p15}{$b$}&\dotnode{p16}{$b $}
				\\   
				\dotnode{p21}{$c $}&\dotnode{p22}{$a $}&\dotnode{p23}{$b $}&\dotnode{p24}{$d$}&\dotnode{p25}{$a$}&\dotnode{p26}{$b $}
				\\   
				\dotnode{p31}{$a $}&\dotnode{p32}{$c $}&\dotnode{p33}{$d $}&\dotnode{p34}{$b $}&\dotnode{p35}{$c$}&\dotnode{p36}{$d $}
				\\
				\dotnode{p41}{$\blue a$}&\dotnode{p42}{$c$}&\dotnode{p43}{$c$}&\dotnode{p44}{$d$}&\dotnode{p45}{$d$}&\dotnode{p46}{$\blue b$} 
				\\
				\dotnode{p51}{$\blue c$}&\dotnode{p52}{$a $}&\dotnode{p53}{$a$}&\dotnode{p54}{$b $}&\dotnode{p55}{$b$}&\dotnode{p56}{$\blue d $}
				\\   
				\dotnode{p61}{$c $}&\dotnode{p62}{$a $}&\dotnode{p63}{$b $}&\dotnode{p64}{$d$}&\dotnode{p65}{$a$}&\dotnode{p66}{$b $}
				\\   
				\dotnode{p71}{$a $}&\dotnode{p72}{$c $}&\dotnode{p73}{$d $}&\dotnode{p74}{$b $}&\dotnode{p75}{$c$}&\dotnode{p76}{$d $}
				\\
				\dotnode{p81}{$c$}&\dotnode{p82}{$c$}&\dotnode{p83}{$c$}&\dotnode{p84}{$d$}&\dotnode{p85}{$d$}&\dotnode{p86}{$d$}

				\ncarc[arcangle=30, linestyle=dashed,linecolor=red,linewidth=0.5pt]{p11}{p16}                  
				\ncarc[arcangle=-30, linestyle=dashed,linecolor=red,linewidth=0.5pt]{p11}{p21}                  
				\ncarc[arcangle=30, linestyle=dashed,linecolor=red,linewidth=0.5pt]{p13}{p14}
				\ncarc[arcangle=-30, linestyle=dashed,linecolor=red,linewidth=0.5pt]{p13}{p43}
				\ncarc[arcangle=30, linestyle=dashed,linecolor=red,linewidth=0.5pt]{p14}{p24}
				\ncarc[arcangle=30, linestyle=solid,linecolor=blue,linewidth=0.5pt]{p16}{p86}\ncput*{junction}
				
				\ncarc[arcangle=-30, linestyle=dashed,linecolor=red,linewidth=0.9pt]{p21}{p24}
				
				\ncarc[arcangle=30, linestyle=dashed,linecolor=red,linewidth=0.5pt]{p31}{p34}
				\ncarc[arcangle=-45, linestyle=solid,linecolor=blue,linewidth=0.5pt]{p31}{p61}\ncput*{junction}
				\ncarc[arcangle=30, linestyle=dashed,linecolor=red,linewidth=0.5pt]{p34}{p44} 
				\ncarc[arcangle=-15, linecolor=blue,linewidth=0.5pt]{p41}{p46}
				\ncarc[arcangle=15, linecolor=blue,linewidth=0.5pt]{p51}{p56}
				\ncarc[arcangle=0, linecolor=blue,linewidth=0.5pt]{p41}{p51}
				\ncarc[arcangle=0, linecolor=blue,linewidth=0.5pt]{p46}{p56}
				
				\ncarc[arcangle=-30, linestyle=dashed,linecolor=red,linewidth=0.5pt]{p43}{p44}  
				
				\ncarc[arcangle=30, linestyle=dashed,linecolor=red,linewidth=0.5pt]{p53}{p54}
				\ncarc[arcangle=-30, linestyle=dashed,linecolor=red,linewidth=0.5pt]{p53}{p83}
				\ncarc[arcangle=30, linestyle=dashed,linecolor=red,linewidth=0.5pt]{p54}{p64}
				
				\ncarc[arcangle=-30, linestyle=dashed,linecolor=red,linewidth=0.9pt]{p61}{p64}
				
				\ncarc[arcangle=30, linestyle=dashed,linecolor=red,linewidth=0.5pt]{p71}{p74}
				\ncarc[arcangle=-30, linestyle=dashed,linecolor=red,linewidth=0.5pt]{p71}{p81}
				\ncarc[arcangle=30, linestyle=dashed,linecolor=red,linewidth=0.5pt]{p74}{p84} 
				
				\ncarc[arcangle=-30, linestyle=dashed,linecolor=red,linewidth=0.5pt]{p81}{p86}
				\ncarc[arcangle=-30, linestyle=dashed,linecolor=red,linewidth=0.5pt]{p83}{p84}

			\end{psmatrix}
		}%
	\end{center}
	\caption{Left. Picture $p_{(1)}$ used as induction basis of Theorem~\ref{theorDCwithUnboundedCircuits}. It is covered by a circuit of length $4+8\cdot 1=12$ and by 3 rectangular circuits. 
		Middle. Picture $p_{(1)}\ominus p_{(1)}$, the four arcs to be deleted are in green, and the four nodes to be relabeled are in blue.
		Right. Inductive step: picture $p_{(2)}$ is obtained from $p_{(1)}\ominus p_{(1)}$ by canceling the four green arcs, relabeling the four blue nodes as shown (the corresponding rectangular circuit is in blue) and 
		finally adding two arcs (blue) that join the double-noose circuits. A circuit of length $4+ 8\cdot 2$ results. Notice that all length 4 circuits of $p_{(h-1)}$ and $p_{(1)}$ are unchanged in $p_{(h)}$.}~\label{figExtensibleLengthCircuitFamily}
\end{figure}

\begin{theorem}{\rm(Unbounded circuit length)}\label{theorDCwithUnboundedCircuits}
	For all $h \ge 0$ there exist a picture in $DC_k$ that contains a circuit of length $4+8h$.
\end{theorem}
\begin{proof}
We prove the statement for $DC_1$, the general case being analogous. The case $h=0$ is obvious. The case $h>0$ is proved by induction on a sequence of pictures $p_{(1)},\ldots p_{(h)}$ using as basis the $DC_1$ picture $p_{(1)}$ in Figure \ref{figExtensibleLengthCircuitFamily} (left), that has size $(m_{(1)}, 6)$, where $m_{(1)}=4$, and contains a circuit of length $12=4+8$, referred to as double-noose. 

\par\noindent 
Induction step. It extends picture $p_{(h-1)}$, $h>1$, by appending a copy of $p_{(1)}$ underneath and making a few changes defined in Figure \ref{figExtensibleLengthCircuitFamily} (right). It is easy to see that the result is a picture $p_{(h)}$ of size $(m_{(h-1)}+4, 6)$ such that: $p_{(h)} \in DC_1$ and $p_{(h)}$ contains a circuit of length $4 + 8h$.
\end{proof}

	Another series of pictures that can be enlarged indefinitely is the one in Figure~\ref{figPictureFarfallaInDC}, where the first two terms of the series are shown. 

\par

\subsection{Quaternate Dyck crosswords}
The next definition forbids any cycle longer than 4 and keeps, e.g., the pictures in Figures~\ref{fig-rectangles} and \ref{figOverlappingRectangles}.
\begin{definition}[Quaternate $DC_k$]\label{defQuaternateDyckLang}
	A Dyck crossword picture such that all its circuits are of length 4 is called \emph{quaternate}; their language, denoted by $DQ_k$, is the \emph{quaternate Dyck language}.
\end{definition}
\begin{figure}[htb]
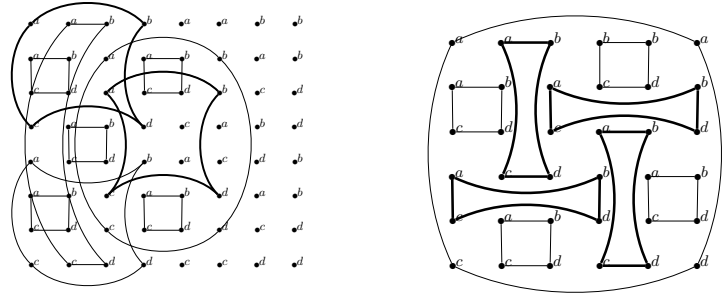

	\begin{center}
		\scalebox{0.50}{%
			\psset{arrows=-,labelsep=3pt,colsep=20pt,rowsep=15pt,nodealign=true}         
			\begin{psmatrix}                                  
				\\                                            
				\dotnode{p11}{$a $}&\dotnode{p12}{$a $}&\dotnode{p13}{$b $}&\dotnode{p14}{$b $} & 
				\dotnode{p15}{$a $}&\dotnode{p16}{$a $}&\dotnode{p17}{$b $}&\dotnode{p18}{$b $}     
				\\   
				\dotnode{p21}{$a$}&\dotnode{p22}{$b$}&\dotnode{p23}{$a $}&\dotnode{p24}{$a$} & 
				\dotnode{p25}{$b$}&\dotnode{p26}{$b $}&\dotnode{p27}{$a$}&\dotnode{p28}{$b$}
				\\
				\dotnode{p31}{$c$}&\dotnode{p32}{$d$}&\dotnode{p33}{$a $}&\dotnode{p34}{$c$} & 
				\dotnode{p35}{$d$}&\dotnode{p36}{$b $}&\dotnode{p37}{$c$}&\dotnode{p38}{$d$}
				\\
				\dotnode{p41}{$c $}&\dotnode{p42}{$a$}&\dotnode{p43}{$b$}&\dotnode{p44}{$d $} & 
				\dotnode{p45}{$c $}&\dotnode{p46}{$a$}&\dotnode{p47}{$b$}&\dotnode{p48}{$d $}
				\\
				\dotnode{p51}{$a $}&\dotnode{p52}{$c$}&\dotnode{p53}{$d$}&\dotnode{p54}{$b $} & 
				\dotnode{p55}{$a $}&\dotnode{p56}{$c$}&\dotnode{p57}{$d$}&\dotnode{p58}{$b $}
				\\
				\dotnode{p61}{$a$}&\dotnode{p62}{$b$}&\dotnode{p63}{$c $}&\dotnode{p64}{$a$} & 
				\dotnode{p65}{$b$}&\dotnode{p66}{$d $}&\dotnode{p67}{$a$}&\dotnode{p68}{$b$}
				\\
				\dotnode{p71}{$c$}&\dotnode{p72}{$d$}&\dotnode{p73}{$c $}&\dotnode{p74}{$c$} & 
				\dotnode{p75}{$d$}&\dotnode{p76}{$d $}&\dotnode{p77}{$c$}&\dotnode{p78}{$d$}
				\\
				\dotnode{p81}{$c $}&\dotnode{p82}{$c $}&\dotnode{p83}{$d $}&\dotnode{p84}{$d $} & 
				\dotnode{p85}{$c $}&\dotnode{p86}{$c $}&\dotnode{p87}{$d $}&\dotnode{p88}{$d $} 
				
				\ncarc[arcangle=45, linestyle=solid,linewidth=1.5pt]{p11}{p14}              
				\ncarc[arcangle=-45, linestyle=solid,linewidth=1.5pt]{p11}{p41}            
				\ncarc[arcangle=0, linestyle=solid,linewidth=0.5pt]{p12}{p13}           
				\ncarc[arcangle=-45, linestyle=solid,linewidth=0.5pt]{p12}{p82}
				\ncarc[arcangle=-45, linestyle=solid,linewidth=1.5pt]{p14}{p44}
				\ncarc[arcangle=-45, linestyle=solid,linewidth=0.5pt]{p13}{p83}                                           
				\ncarc[arcangle=0, linestyle=solid,linewidth=0.5pt]{p21}{p22}             
				\ncarc[arcangle=0, linestyle=solid,linewidth=0.5pt]{p21}{p31}
				\ncarc[arcangle=0, linestyle=solid,linewidth=0.5pt]{p22}{p32}             
				\ncarc[arcangle=45, linestyle=solid,linewidth=0.5pt]{p23}{p26}             
				\ncarc[arcangle=-45, linestyle=solid,linewidth=0.5pt]{p23}{p73}
				\ncarc[arcangle=-0, linestyle=solid,linewidth=0.5pt]{p24}{p25}
				\ncarc[arcangle=-0, linestyle=solid,linewidth=0.5pt]{p24}{p34}
				\ncarc[arcangle=-0, linestyle=solid,linewidth=0.5pt]{p25}{p35}
				\ncarc[arcangle=45, linestyle=solid,linewidth=0.5pt]{p26}{p76}
				
				\ncarc[arcangle=0, linestyle=solid,linewidth=0.5pt]{p31}{p32}
				\ncarc[arcangle=45, linestyle=solid,linewidth=1.5pt]{p33}{p63}
				\ncarc[arcangle=45, linestyle=solid,linewidth=1.5pt]{p33}{p36}
				\ncarc[arcangle=0, linestyle=solid,linewidth=0.5pt]{p34}{p35} 
				\ncarc[arcangle=-45, linestyle=solid,linewidth=1.5pt]{p36}{p66}
				
				\ncarc[arcangle=45, linestyle=solid,linewidth=1.5pt]{p41}{p44}
				\ncarc[arcangle=0, linestyle=solid,linewidth=0.5pt]{p42}{p43}
				\ncarc[arcangle=0, linestyle=solid,linewidth=0.5pt]{p42}{p52}
				\ncarc[arcangle=0, linestyle=solid,linewidth=0.5pt]{p43}{p53}
				
				\ncarc[arcangle=-45, linestyle=solid,linewidth=0.5pt]{p51}{p54}              
				\ncarc[arcangle=-45, linestyle=solid,linewidth=0.5pt]{p51}{p81}            
				\ncarc[arcangle=0, linestyle=solid,linewidth=0.5pt]{p52}{p53}           
				\ncarc[arcangle=-45, linestyle=solid,linewidth=0.5pt]{p54}{p84}
				
				\ncarc[arcangle=0, linestyle=solid,linewidth=0.5pt]{p61}{p62}             
				\ncarc[arcangle=0, linestyle=solid,linewidth=0.5pt]{p61}{p71}
				\ncarc[arcangle=0, linestyle=solid,linewidth=0.5pt]{p62}{p72}
				\ncarc[arcangle=45, linestyle=solid,linewidth=1.5pt]{p63}{p66}
				\ncarc[arcangle=0, linestyle=solid,linewidth=0.5pt]{p64}{p65}
				\ncarc[arcangle=0, linestyle=solid,linewidth=0.5pt]{p64}{p74}
				\ncarc[arcangle=0, linestyle=solid,linewidth=0.5pt]{p65}{p75}
				
				\ncarc[arcangle=0, linestyle=solid,linewidth=0.5pt]{p71}{p72}
				\ncarc[arcangle=-45, linestyle=solid,linewidth=0.5pt]{p73}{p76}
				\ncarc[arcangle=45, linestyle=solid,linewidth=0.5pt]{p33}{p36}
				\ncarc[arcangle=0, linestyle=solid,linewidth=0.5pt]{p74}{p75} 
				
				\ncarc[arcangle=-45, linestyle=solid,linewidth=0.5pt]{p81}{p84} 
				\ncarc[arcangle=0, linestyle=solid,linewidth=0.5pt]{p82}{p83}            
			\end{psmatrix}
			
		}%
		\qquad\qquad\qquad
		\scalebox{0.65}{%
			\psset{arrows=-,labelsep=3pt,colsep=20pt,rowsep=15pt,nodealign=true}         
			\begin{psmatrix}                                  
				\\                                            
				\dotnode{p11}{$a $}&\dotnode{p12}{$a $}&\dotnode{p13}{$b $}&\dotnode{p14}{$b $} & 
				\dotnode{p15}{$b $}&\dotnode{p16}{$a $}     
				\\   
				\dotnode{p21}{$a$}&\dotnode{p22}{$b$}&\dotnode{p23}{$a $}&\dotnode{p24}{$c$} & 
				\dotnode{p25}{$d$}&\dotnode{p26}{$b $}
				\\
				\dotnode{p31}{$c$}&\dotnode{p32}{$d$}&\dotnode{p33}{$c $}&\dotnode{p34}{$a$} & 
				\dotnode{p35}{$b$}&\dotnode{p36}{$d $}
				\\
				\dotnode{p41}{$a $}&\dotnode{p42}{$c$}&\dotnode{p43}{$d$}&\dotnode{p44}{$b $} & 
				\dotnode{p45}{$a $}&\dotnode{p46}{$b$}
				\\
				\dotnode{p51}{$c $}&\dotnode{p52}{$a$}&\dotnode{p53}{$b$}&\dotnode{p54}{$d$} & 
				\dotnode{p55}{$c $}&\dotnode{p56}{$d$}
				\\
				\dotnode{p61}{$c$}&\dotnode{p62}{$c$}&\dotnode{p63}{$d $}&\dotnode{p64}{$c$} & 
				\dotnode{p65}{$d$}&\dotnode{p66}{$d $}
				\ncarc[arcangle=25, linestyle=solid,linewidth=0.5pt]{p11}{p16} 
				\ncarc[arcangle=25, linestyle=solid,linewidth=0.5pt]{p16}{p66}		
				\ncarc[arcangle=25, linestyle=solid,linewidth=0.5pt]{p66}{p61}
				\ncarc[arcangle=25, linestyle=solid,linewidth=0.5pt]{p61}{p11}

				\ncarc[arcangle=0, linestyle=solid,linewidth=0.5pt]{p14}{p15} 
				\ncarc[arcangle=0, linestyle=solid,linewidth=0.5pt]{p15}{p25}				
				\ncarc[arcangle=0, linestyle=solid,linewidth=0.5pt]{p25}{p24}
				\ncarc[arcangle=0, linestyle=solid,linewidth=0.5pt]{p24}{p14} 
				
				\ncarc[arcangle=0, linestyle=solid,linewidth=0.5pt]{p21}{p22} 
				\ncarc[arcangle=0, linestyle=solid,linewidth=0.5pt]{p22}{p32}				
				\ncarc[arcangle=0, linestyle=solid,linewidth=0.5pt]{p32}{p31}
				\ncarc[arcangle=0, linestyle=solid,linewidth=0.5pt]{p31}{p21} 
				
				\ncarc[arcangle=0, linestyle=solid,linewidth=0.5pt]{p45}{p46} 
				\ncarc[arcangle=0, linestyle=solid,linewidth=0.5pt]{p46}{p56}				
				\ncarc[arcangle=0, linestyle=solid,linewidth=0.5pt]{p56}{p55}
				\ncarc[arcangle=0, linestyle=solid,linewidth=0.5pt]{p55}{p45} 
				
				\ncarc[arcangle=0, linestyle=solid,linewidth=0.5pt]{p52}{p53} 
				\ncarc[arcangle=0, linestyle=solid,linewidth=0.5pt]{p53}{p63}				
				\ncarc[arcangle=0, linestyle=solid,linewidth=0.5pt]{p63}{p62}
				\ncarc[arcangle=0, linestyle=solid,linewidth=0.5pt]{p62}{p52} 
				\ncarc[arcangle=0,linestyle=solid,linewidth=1.5pt]{p12}{p13}           \ncarc[arcangle=25,linestyle=solid,linewidth=1.5pt]{p12}{p42}				\ncarc[arcangle=0,linestyle=solid,linewidth=1.5pt]{p42}{p43} 				\ncarc[arcangle=25,linestyle=solid,linewidth=1.5pt]{p43}{p13}               				
				
				\ncarc[arcangle=-25, linestyle=solid,linewidth=1.5pt]{p23}{p26} 
				\ncarc[arcangle=0, linestyle=solid,linewidth=1.5pt]{p26}{p36}
				\ncarc[arcangle=-25, linestyle=solid,linewidth=1.5pt]{p36}{p33}
				\ncarc[arcangle=0, linestyle=solid,linewidth=1.5pt]{p33}{p23} 				
				
				\ncarc[arcangle=-25, linestyle=solid,linewidth=1.5pt]{p41}{p44} 
				\ncarc[arcangle=0, linestyle=solid,linewidth=1.5pt]{p44}{p54}
				\ncarc[arcangle=-25, linestyle=solid,linewidth=1.5pt]{p54}{p51}
				\ncarc[arcangle=0, linestyle=solid,linewidth=1.5pt]{p51}{p41}
				
				\ncarc[arcangle=0,linestyle=solid,linewidth=1.5pt]{p34}{p35}           \ncarc[arcangle=-25,linestyle=solid,linewidth=1.5pt]{p35}{p65}				\ncarc[arcangle=0,linestyle=solid,linewidth=1.5pt]{p65}{p64} 				\ncarc[arcangle=-25,linestyle=solid,linewidth=1.5pt]{p64}{p34}

			\end{psmatrix}
			
		}
	\end{center}
	
	\caption{(Left). A quaternate picture (left) with two overlapping rectangles (thicker lines) that mutually include only one node of the other. To avoid clogging, the rectangles in the specular right half of the picture are not drawn. Such a picture is not neutralizable (Definition \ref{def2DDyckByNeutralization}). The precedence relation (Definition~\ref{defPrecedenceRelation}) is not acyclic since $(1,1)\prec (3,3) \prec (1,1)$, where each rectangle is identified by the coordinate of its north-west node. Another quaternate picture (right) shows a cycle of length 4: $(1,2)\prec (4,1)\prec (3,4)\prec (2,3)\prec (1,2)$.}
	\label{figOverlappingRectangles}
\end{figure}

\section{Language inclusions}~\label{s-inclusions}
In this section we show the strict language inclusions existing between the alternative definitions of 2D Dyck languages.
\par
Since $DC_k$ pictures may contain circuits of length $>4$, (e.g., in Figure~\ref{figPictureFarfallaInDC}) quaternate Dyck languages are strictly included in Dyck crosswords.
\par It is obvious that $DN_k\subseteq DQ_k$; a natural question is then whether the inclusion is strict. 
To answer, we define a precedence relation between two rectangles of a $DQ_k$ picture such that the first must be neutralized before the second.
\begin{definition}[precedence in neutralization]\label{defPrecedenceRelation}
Let $p\in DQ_k$ and let $\alpha$ and $\beta$ two rectangles (i.e. length 4 circuits) occurring in $p$. Rectangle $\alpha$ has {\em priority} over $\beta$ 
if, and only if, one, two or four nodes of $\alpha$ fall inside rectangle $\beta$ or on its sides. (For three nodes it is impossible.). Let $\prec$, the {\em precedence relation}, be the transitive closure of the priority relation.
\end{definition}
\begin{example}[precedence relation]\label{exPrecedenceRela}
The precedence relation for the picture in Figure~\ref{figOverlappingRectangles}, left, has the length-2 cycle 
$(1,1) \prec (3,3) \prec (1,1)$, blocking the neutralization process of the two rectangles evidenced by thicker lines.
The picture in Figure~\ref{figOverlappingRectangles}, right, has a cycle of length 4.
\end{example}

\begin{theorem}[neutralizable vs quaternate]\label{theorNeutralizableVSquaternate}
A picture in $DQ_k$ is neutralizable if and only if its precedence relation is acyclic. 
\end{theorem}
\begin{proof}
Let relation $\prec$ be acyclic. Then sort the rectangles in topological order and apply neutralization starting from a rectangle without predecessors. When a rectangle is checked, all of its predecessors have already been neutralized, and neutralization can proceed until all rectangles are neutralized.
The converse is obvious: if relation $\prec$ has a cycle no rectangle in the cycle can be neutralized.
\end{proof}

\par
From previous properties of the various 2D Dyck languages introduced in this paper, we obtain a strict linear  hierarchy with respect to language inclusion.

\begin{corollary}[hierarchy]\label{CoroHierarchy}
	$DW_k \subsetneq DN_k \subsetneq DQ_k\subsetneq DC_k.$
\end{corollary}

\section{Conclusion}\label{s-conclusion}
By introducing some definitions of 2D Dyck languages  we have made the first step towards a new characterization of 2D context-free languages by means of the Chomsky-Sch\"{u}tzenberger theorem suitably reformulated for picture languages. But, in our opinion, the mathematical study of the properties of 2D Dyck languages has independent interest, and much remains to be understood, especially for the richer case of Dyck crosswords.
Very diverse patterns may occur in $DC_k$ pictures, that currently we are unable to classify. The variety of patterns is related to the length of the circuits in the matching graph and to the number of intersection points in a circuit or between different circuits.
\par
We mention two specific open problems. 
(i) The picture ${\genfrac{}{}{0pt}{2}{ab}{cd}}$ has just one circuit, which is therefore Hamiltonian; it is not known whether there exist any other Hamiltonian pictures in $DC_1$. 
(ii) By Theorem~\ref{theorDCwithUnboundedCircuits} the length of circuits in $DC_1$ pictures is unbounded. The question is whether, for all values $n>1$, there is a $DC_1$ picture containing a circuit of length $4n$.
\par
{A related range of questions concerns the "productivity" of a circuit, meaning the existence of $DC_k$ pictures incorporating the circuit. A simple formulation is: given a circuit situated in its bounding box, does a $DC_k$ picture exist of a size equal or larger than the bounding box, such that the same circuit occurs within the picture? 
}

\textbf{Acknowledgment}: We thank Matteo Pradella for helpful discussions.
\bibliographystyle{abbrv}
\bibliography{automatabib}

\begin{thebibliography}{10}

\bibitem{Berstel79}
J.~Berstel.
\newblock {\em Transductions and Context-Free Languages}.
\newblock Teubner, Stuttgart, 1979.

\bibitem{DBLP:journals/fuin/BerstelB96}
J.~Berstel and L.~Boasson.
\newblock Towards an algebraic theory of context-free languages.
\newblock {\em Fundam. Informaticae}, 25(3):217--239, 1996.

\bibitem{ChomskySchutz1963}
N.~Chomsky and M.~Sch{\"u}tzenberger.
\newblock The algebraic theory of context-free languages.
\newblock In Brafford and Hirschenber, editors, {\em Computer programming and
  formal systems}, pages 118--161. North-Holland, Amsterdam, 1963.

\bibitem{DBLP:books/ems/21/Crespi-ReghizziGL21}
S.~{Crespi Reghizzi}, D.~Giammarresi, and V.~Lonati.
\newblock Two-dimensional models.
\newblock In J.~Pin, editor, {\em Handbook of Automata Theory}, pages 303--333.
  European Mathematical Society Publishing House, 2021.

\bibitem{DBLP:journals/tcs/Crespi-ReghizziP05}
S.~Crespi{-}Reghizzi and M.~Pradella.
\newblock Tile rewriting grammars and picture languages.
\newblock {\em Theor. Comput. Sci.}, 340(1):257--272, 2005.

\bibitem{Drewes:2006}
F.~Drewes.
\newblock {\em Grammatical Picture Generation: {A} Tree-Based Approach}.
\newblock Springer, 2006.

\bibitem{DBLP:journals/iandc/FennerPT22}
S.~A. Fenner, D.~Pad{\'{e}}, and T.~Thierauf.
\newblock The complexity of regex crosswords.
\newblock {\em Inf. Comput.}, 286:104777, 2022.

\bibitem{GiammRestivo1997}
D.~Giammarresi and A.~Restivo.
\newblock Two-dimensional languages.
\newblock In G.~Rozenberg and A.~Salomaa, editors, {\em Handbook of formal
  languages, vol. 3}, pages 215--267. Springer, 1997.

\bibitem{DBLP:journals/tcs/LatteuxS97}
M.~Latteux and D.~Simplot.
\newblock Recognizable picture languages and domino tiling.
\newblock {\em Theor. Comput. Sci.}, 178(1-2):275--283, 1997.

\bibitem{Matz:1997}
O.~Matz.
\newblock Regular expressions and context-free grammars for picture languages.
\newblock In {\em 14th Annual Symposium on Theoretical Aspects of Computer
  Science}, volume 1200 of {\em LNCS}, pages 283--294, 1997.

\bibitem{Nivat91SSSD}
M.~Nivat, A.~Saoudi, K.~G. Subramanian, R.~Siromoney, and V.~R. Dare.
\newblock Puzzle grammars and context-free array grammars.
\newblock {\em Int. Journ. of Pattern Recognition and Artificial Intelligence},
  5:663--676, 1991.

\bibitem{okhotin2012}
A.~Okhotin.
\newblock Non-erasing variants of the {Chomsky---Sch\"{u}tzenberger Theorem}.
\newblock In {\em Proc. of the 16th Intern. Conf. on Developments in Language
  Theory}, DLT'12, pages 121--129, Berlin, Heidelberg, 2012. Springer-Verlag.

\bibitem{DBLP:journals/pr/PradellaC08}
M.~Pradella and S.~Crespi{-}Reghizzi.
\newblock A sat-based parser and completer for pictures specified by tiling.
\newblock {\em Pattern Recognit.}, 41(2):555--566, 2008.

\bibitem{Prusa2004}
D.~Pr{\r{u}}{\v s}a.
\newblock {\em Two-dimensional Languages ({PhD Thesis})}.
\newblock Charles University, Faculty of Mathematics and Physics, Czech
  Republic, 2004.

\bibitem{DBLP:journals/tcs/Simplot99}
D.~Simplot.
\newblock A characterization of recognizable picture languages by tilings by
  finite sets.
\newblock {\em Theor. Comput. Sci.}, 218(2):297--323, 1999.

\bibitem{Siromoney&Subramanian&Rajkumar&Thomas:1999}
R.~Siromoney, K.~G. Subramanian, V.~R. Dare, and D.~G. Thomas.
\newblock Some results on picture languages.
\newblock {\em Pattern Recognition}, 32(2):295--304, 1999.

\end{thebibliography}

\end{document}